\documentclass{article}

\usepackage[english]{babel}

\usepackage[letterpaper,top=2cm,bottom=2cm,left=3cm,right=3cm,marginparwidth=1.75cm]{geometry}

\usepackage{amsmath}
\usepackage[colorlinks=true, allcolors=blue]{hyperref}

\usepackage{pre_all}
\usepackage{pre_CS}
\usepackage[blocks]{authblk}

\usepackage{accents}

\theoremstyle{remark}

\newcommand\form[1]{\underaccent{\vec}{#1}}

\newcommand{\FF}{{\mathcal F}}
\let\C\relax
\newcommand{\C}{{\mathbb C}}
\newcommand{\ZZ}{{\mathbb Z}}

\newcommand{\Iota}{{\mathcal B}}

\newcommand{\coh}{{\rm coh}}
\newcommand{\hilb}{{\mathfrak H}}

\newcommand{\Stat}{{\mathcal St}}

\newcommand{\ro}{ro}
\DeclareMathOperator{\Aff}{Aff}

\DeclareMathOperator{\Sp}{Sp}
\DeclareMathOperator{\Met}{Met}
\newcommand{\met}{\Theta}

\newcommand{\SLCC}{\SL(2,\C)}
\newcommand{\slcc}{\operatorname{sl}(2,\C)}
\DeclareMathOperator{\Hol}{Hol}
\DeclareMathOperator{\Ad}{Ad}

\newcommand{\Blue}{}
\newcommand{\blue}{}

\title{Hessian in the spinfoam models with cosmological constant}

\author[1]{{\bf Wojciech Kami{\'n}ski}}
\affil[1]{Faculty of Physics, University of Warsaw,
ul. Pasteura 5, 02-093 Warsaw, Poland}
\author[2,3]{{\bf Qiaoyin Pan\textsuperscript{$\ast$}}}
\affil[2]{Yau Mathematical Sciences Center, Tsinghua University, Beijing 100084, China}
\affil[3]{Department of Physics, Florida Atlantic University, 777 Glades Road, Boca Raton, FL 33431, USA}

\begin{document}
\maketitle

\begin{abstract}
In this paper, we introduce a general method to prove the non-degeneracy of the Hessian in the spinfoam vertex amplitude for quantum gravity and apply it to the spinfoam models with a cosmological constant ($\Lambda$-SF model).  By reformulating the problem in terms of the transverse intersection of some submanifolds in the phase space of flat $\SL(2,\bC)$ connections, we demonstrate that the Hessian is non-degenerate for critical points corresponding to non-degenerate, geometric 4-simplices in de Sitter or anti-de Sitter space.  Non-degeneracy of the Hessian is an important necessary condition for the stationary phase method to be applicable. With a non-degenerate Hessian, this method not only confirms the connection of the $\Lambda$-SF model to semiclassical gravity, but also shows that there are no dominant contributions from exceptional configurations as in the Barrett-Crane model. Given its general nature, we expect our criterion to be applicable to other spinfoam models under mild adjustments. 
\end{abstract}
\tableofcontents

\section{\blue{Introduction}}

Spinfoam models are covariant formulations of Loop Quantum Gravity \cite{Rovelli:2014ssa, Perez:2012wv,Baez:1999sr,Bianchi:2017loop,Engle:2023qsu,Livine:2024hhc}. {\Blue Their amplitudes are defined as a sum over labels of the product of composite amplitudes for a given triangulation of the manifold. The precise definition of a state sum depends on the specific details of the model. It can be introduced as a limit over a refinement of triangulations \cite{Rovelli:2010qx}, sum over possible triangulations of a manifold  \cite{Han:2025emp}, or, in the Group Field Theory approach, a sum over triangulations of all possible manifolds with given boundary} \cite{DePietri:1999bx,Reisenberger:2000zc,Freidel:2005qe,Oriti:2006se}. 

The fundamental building block of these theories is the  vertex amplitude that attracted much attention. In particular, the semiclassical relation to quantum gravity relies crucially on the asymptotic behavior of this amplitude. {\Blue The relation of spinfoam models, which are defined using group (or quantum group) theory, to Einstein's gravity has a long history.  It started with the seminal observation by Ponzano and Regge relating the asymptotics of $6j$-symbols to 3D gravity \cite{Ponzano:1968se}, a result subsequently proven through different methods \cite{Roberts:1998zka,Schulten:1975sem,Freidel:2002mj,Barrett:1993db}. This correspondence was later extended to include a cosmological constant \cite{taylor20066}, linking the asymptotics of $\SU_q(2)$ $6j$-symbol to 3D gravity with a cosmological constant.  }

\blue{In the four-dimensional context, the first significant development was the Barrett-Crane model \cite{Barrett:1997gw,Barrett:1999qw}, the asymptotics of which were analyzed in \cite{Barrett:1997gw,Barrett:1998gs,Barrett:1999qw, Freidel:2002mj}. Fundamental result of \cite{Barrett:2011xa, Barrett:2009mw} demonstrated that the asymptotics of the 4D EPRL vertex amplitude recover the Einstein-Hilbert action on a suitably constructed 4-simplex. This analysis confirmed that spinfoam models could successfully describe gravitational degrees of freedom in the semiclassical limit.}

\blue{Developing a physical 4D theory requires navigating several distinctions in model construction. While 3D models are topological, 4D models must capture local degrees of freedom. Furthermore, while Euclidean models offer mathematical simplicity, Lorentzian models are necessary to describe gravity with the correct physical signature, though they introduce the complexities of the non-compact Lorentz group. A third, critical distinction, and the focus of this work, lies in the treatment of the cosmological constant, $\Lambda$. We distinguish between ``flat models" ($\Lambda=0$) and ``$\Lambda$-models" ($\Lambda\neq 0$). Flat models, defined using classical groups, typically suffer from infrared divergences requiring renormalization \cite{Freidel:2004vi,Bonzom:2009zd,Bonzom:2012mb,Goeller:2019zpz,Barrett:2008wh}}\footnote{\blue{Let us note that the na\"ive divergence of the Ponzano-Regge model (3D Euclidean spinfoam without a cosmological constant) has been well understood \cite{Freidel:2004vi,Bonzom:2009zd,Goeller:2019zpz}. It is related to the translation group of the internal vertices of a triangulation, whose volume is infinite. It can be regularized by a gauge-fixing process \cite{Freidel:2004vi,Bonzom:2012mb} or using the cohomological criterion \cite{Barrett:2008wh}, leading to well-defined topological invariants. }}.
\blue{In contrast, $\Lambda$-models are typically defined in terms of quantum groups or Chern-Simons theory \cite{Turaev:1992hq,Noui:2002ag,Han:2010pz,Fairbairn:2010cp,Han:2021tzw,Han:2025mkc}, where $\Lambda$ is treated as a deformation parameter or a coupling constant}.

The inclusion of a cosmological constant has various motivations. For specific values of $\Lambda\not=0$, the labels in a sum turn out to be of a finite range and the model does not need renormalization. 
This happens in the case of the Turaev-Viro model in 3D \cite{Turaev:1992hq}, \blue{resulting from the use of quantum group $\SU_q(2)$ at the root of unity}, and in Lorentzian 4D models considered in this paper \cite{Han:2021tzw,Han:2025mkc} \blue{as a result of the finite Chern-Simons level}. Additionally, observations \cite{SupernovaSearchTeam:1998fmf,SupernovaCosmologyProject:1998vns,Planck:2018vyg} indicate the presence of a cosmological constant in our universe, motivating the study of 4D Lorentzian spinfoam with $\Lambda\neq 0$.

However, there are subtleties in the definition of the models for a non-zero $\Lambda$. It was proposed quite early \cite{Baez:1999sr} that the \blue{4D spinfoam} models \blue{with a $\Lambda\neq 0$} should be related to Chern-Simons theory \blue{from the formal path integral formalism of the BF theory with a $B\w B$ term} \cite{Horowitz:1989ng,Cotta-Ramusino:1994nhr,Cattaneo:1995tw}\footnotemark{}, but explicit constructions were plagued by some problems. \blue{The initial attempts \cite{Noui:2002ag,Han:2010pz,Fairbairn:2010cp} employed $q$-deformation of flat models. They are manifestly finite, but their semiclassical limit is complicated to study. In contrast, the models developed in \cite{Haggard:2015sl,Haggard:2015yda,Haggard:2015nat} maintain good control over the semiclassical regime. 
However, because they rely on a formal path integral formulation, they lack a rigorous mathematical definition.}

In this paper, we concentrate on a recent development \cite{Han:2021tzw, Han:2023hbe, Han:2024reo, Han:2025mkc, Pan:2025sut} based on $\SLCC$ Chern-Simons theory, as formulated in a series of works \cite{Gaiotto:2009hg,Dimofte:2011gm,Dimofte:2011ju,Dimofte:2013lba,Dimofte:2014zga,andersen2014complex}. We will call this spinfoam model the $\Lambda$-SF model. \blue{The most outstanding benefit of the $\Lambda$-SF model is that the amplitude is not only finite by construction but also recovers 4D Lorentzian Regge action with a comological constant in the semiclassical approximation. Recent developments of the $\Lambda$-SF model have been focusing on the semiclassical behaviors and geometrical interpretation of the amplitude corresponding to different triangulations of the spacetime  \cite{Han:2023hbe,Han:2024reo,Pan:2025sut}, where asymptotic analysis is the key for deriving the results.  }
\footnotetext{\blue{Another relation between 4D quantum gravity and Chern-Simons action from a different point of view is given in \cite{Smolin:1995vq}. It relies on imposing self-dual boundary conditions on a finite boundary in 4D Euclidean gravity, which leads to a boundary term in the action that is exactly the Chern–Simons action. }}

\blue{The asymptotic analysis of spinfoam vertex amplitudes has a well-established history, central to their geometric interpretation. 
For the flat models, it starts from studying the asymptotics of the Euclidean and Lorentzian $10j$-symbol, the building blocks of the Barrett-Crane model \cite{Barrett:1997gw,Barrett:1999qw}. These were analyzed both analytically \cite{Barrett:1998gs, Freidel:2002mj,Barrett:2002ur} and numerically \cite{Baez:2002rx,Baez:2002aw,Christensen:2007rv,Christensen:2009bi}, establishing the initial relation with 4-simplex geometry.
In the more recent EPRL model \cite{Engle:2007wy}, the use of coherent states, developed by Livine and Speziale \cite{Livine:2007vk,Livine:2007ya} and advanced by Freidel and Krasnov \cite{Freidel:2007py}, was crucial in recovering the geometry of a 4-simplex in the semiclassical regime, as shown in detail in \cite{Barrett:2011xa, Barrett:2009mw}.}
Later, the \blue{asymptotic analysis} results were extended to the case with a cosmological constant \cite{Han:2021tzw,Han:2025mkc} (based on geometric results of \cite{Haggard:2015ima} and \cite{Haggard:2015sl}). For the $\Lambda$-models, the $4$-simplices are embedded in constantly curved de Sitter or Anti-de Sitter spacetimes\footnote{In fact, by analysis of Plebanski action, one can show that both signs of cosmological constant should be expected \cite{Haggard:2014xoa}.}. 

In all these cases, the analysis is based on the stationary phase approximation. The method of the stationary phase is a basic tool in asymptotic analysis. It provides an asymptotic expansion based on a few assumptions. One of the important points among these assumptions is the non-degeneracy of the matrix of second derivatives of some action\footnote{\blue{We consider the action after removing gauge freedom by some gauge fixing procedure.}}. We will call this object a Hessian, and it is the main character of our work. 
This condition is more important than it may appear at first sight. In the case when the Hessian is degenerate, the asymptotic of the vertex amplitude is less suppressed. Such configurations are thus expected to be dominant, spoiling the good asymptotic behavior (recovering simplex geometries) of the models. Hence, it is desirable to exclude these pathological behaviors.
\blue{The role of the Hessian has been thoroughly stressed in prior literature, from the analysis of the $10j$-symbol \cite{Barrett:1998gs, Freidel:2002mj,Christensen:2007rv,Christensen:2009bi} to graviton propagator calculations in spinfoams \cite{Alesci:2008ff,Bianchi:2009ri, Bianchi:2011hp}. }

The Hessian in spinfoam models is typically a large matrix, and the determinant is difficult to compute. There are a few results about this object. First, it can be explicitly computed in the case of Ponzano-Regge \cite{Kaminski:2013gaa} and Barrett-Crane models \cite{Kaminski:2013yca}. In the first case, it is non-degenerate \blue{(for geometric configurations that correspond to non-degenerate tetrahedra)}, leading to the known Ponzano-Regge asymptotic \cite{Ponzano:1968se}. In the second case, however, there exist geometric configurations for which the Hessian is degenerate {\Blue \cite{Kaminski:2013yca}. More precisely, the Hessian is degenerate for certain special lengths for which the map from edge lengths to face areas is not a local diffeomorphism (such configurations were described explicitly, for example, in \cite{Dittrich:2007wm})}. In the EPRL model, the computation of the Hessian is virtually impossible, leading to numerical studies \cite{Han:2020fil,Han:2021kll,Han:2023cen}. These studies showed that the Hessian is non-degenerate for generic situation. However, results from the Barrett-Crane model cast doubts on whether the determinant is always non-zero. Quite surprisingly, in \cite{Kaminski:2019dld}, it was proven that the Hessian is non-degenerate at the critical points corresponding to non-degenerate 4-simplices in the EPRL cases (both Lorentzian and Euclidean). 

This leads to a natural question: does the same property hold for the $\Lambda$-SF  model? \blue{Similar to the EPRL case, it is difficult to compute the determinant analytically due to the large size of the Hessian matrix. Non-degeneracy of the Hessian in the $\Lambda$-SF model \cite{Han:2021tzw,Han:2025mkc} was only assumed and passed some numerical tests. Such an assumption continued to be applied in later developments of the models \cite{Han:2023hbe,Han:2024reo,Pan:2025sut}. The good news is}, the method of \cite{Kaminski:2019dld} does not include the actual computation of the Hessian, but only utilizes special properties of the action. \blue{Namely, the action therein satisfies the ``reality conditions'', allowing us to study its geometric theory using the theory of positive Lagrangians introduced in \cite{hormander1973existence}. This gives us the hope that similar techniques could be used in the $\Lambda$-SF model to study the Hessian in generic cases. }

In this paper, we show that, under standard assumptions about boundary states considered in \cite{Han:2021tzw, Han:2025mkc}, the Hessian for the $\Lambda$-SF  model \cite{Han:2021tzw, Han:2025mkc} is non-degenerate. In fact, our result can be divided into two independent parts. The first part concerns restating the condition of non-degeneracy of the Hessian in terms of the intersection of some submanifolds in the space of flat connections. This part is model-dependent, but we expect that every $\Lambda$-model \blue{defined as a constrained topological quantum field theory (TQFT)} will allow for such a reduction. 
\blue{This is based on the natural assumption that, in the semiclassical expansion of the vertex amplitude,  the action can be separated into two parts, one coming from the topological theory itself, and the other coming from the imposition of constraints. These two parts should correspond to different submanifolds of the phase space of the classical topological theory, and the critical points lie in the intersection of these submanifolds, whose feature is related to the Hessian.} The method introduced in this paper is quite general in nature and can be applied to various models. The second part is geometric in nature. It concerns certain properties of holonomies of a constantly curved non-degenerate $4$-simplex together with a well-adapted description of the moduli space of flat connections. The final argument turns out to be very similar to the work in the flat case \cite{Kaminski:2019dld}, but a bit more complex. A difficulty here is a lack of a global frame \blue{(a set of orthonormal basis vectors defined over the entire simplex by the parallel propagation) when a non-zero curvature is present. Therefore, vectors alone cannot encode the geometry consistently, and holonomies are fundamentally necessary} \cite{Haggard:2015ima}.


\subsection*{\blue{Description of the results}}
\label{sec:description}

The main objective of this paper is to extend the analysis of Hessian non-degeneracy from the flat EPRL spin foam model (established in \cite{Kaminski:2019dld}) to the $\Lambda$-SF model with a cosmological constant \cite{Han:2021tzw,Han:2025mkc}. Our main result confirms that the pathological configurations discussed in the introduction are absent in the semiclassical limit of this model. Informally, the result can be stated as follows (see Theorem \ref{thm-1-formal} for the formal statement):

\begin{theorem}\label{thm-1}
The Hessian at the stationary point for the action of \cite{Han:2021tzw, Han:2025mkc} is non-degenerate for any boundary data corresponding to a non-degenerate curved $4$-simplex.
\end{theorem}

Proving this theorem directly is computationally intractable. Unlike the Ponzano-Regge \cite{Kaminski:2013gaa} or Barrett-Crane \cite{Kaminski:2013yca} models, where asymptotic formulas are explicit, the EPRL-type models rely on boundary coherent states. Consequently, the Hessian depends non-trivially on the ``spread" of these states, rendering explicit geometric formulas for the matrix entries essentially impossible to derive. To overcome this, we adopt the strategy introduced in \cite{Kaminski:2019dld}, which avoids direct computation of the determinant by exploiting the structural properties of the action (reality condition). The additional difficulty compared to \cite{Kaminski:2019dld} is that the variables used in the construction of the $\Lambda$-SF model (the Fock-Goncharov and Fensel-Nielsen (FG-FN) coordinates) do not have a direct interpretation in the geometry of the reconstructed curved $4$-simplex \cite{Han:2021tzw,Han:2025mkc}. We overcome this problem by expressing special properties of the Hessian in differential geometry terms of the phase space (without direct reference to the actions).
The logic of our proof proceeds in three distinct steps:

{\bf 1. Reduction to Lagrangian Intersections:} 
First, we utilize the ``reality conditions" of the action. As shown in \cite{Kaminski:2019dld}, a vector lies in the kernel of the Hessian if and only if it is annihilated separately by both the real and imaginary parts of the Hessian. Since the imaginary part is a sum of positive semi-definite matrices, this imposes strong constraints. We elevate this analytical condition to a geometric one by introducing  {\it real Lagrangian parts} (see Definition \ref{df:real-Lagriangian-part}). It is well-known that non-degeneracy of the Hessian in the case of real actions can be expressed by properties of intersections of corresponding Lagrangian submanifolds (see \eg \cite{guillemin2013semi}). With complex actions, one associates so-called positive Lagrangians introduced by Mellin, Sj{\"o}strand and H{\"o}rmander \cite{Mellin-Sjostrand-75,hormander-IV}. Their definition by integrable distributions is harder to handle. However, only small amounts of data of these objects enter the computation of the Hessian. The real Lagrangian parts describe this part of the data encoded in the actions. We demonstrate that non-degeneracy of the Hessian is equivalent to the transverse intersection of two real Lagrangian part submanifolds in the phase space (Propositions \ref{prop:clean-0} and \ref{lm:U-asym}).
The theory of real Lagrangian parts will be developed in Section \ref{sec:stationary_phase}, together with an additional detail of the possible addition of a metaplectic transformation. We expect that this part will have broader application not only to the theory of spin foam models.

{\bf 2. Mapping to the Space of Flat Connections:} 
The phase space of the $\Lambda$-SF model is defined using the FG-FN coordinates. While convenient for quantization, these coordinates lack a direct interpretation in terms of the geometry of the reconstructed 4-simplex. However, there is a natural map from the FG-FN coordinates to the phase space $\cP_\Sigma$ of moduli space of $\SL(2,\bC)$ flat connections on a surface $\Sigma$ introduced in Section \ref{sec:critical_spinfoam}.
We show that the problem can be pushed forward via a local diffeomorphism to $\cP_{\Sigma}$. In this symplectic space, the critical point of the oscillatory integral corresponds to the intersection of two submanifolds: 
 \begin{itemize}
 	\item $\cL_{\rm coh}(m)$: determined by the boundary coherent states labelled by $m$\footnote{
The subscript ``coh" indicates that this submanifold $\cL_{\coh}(m)$ comes from an action encoding coherent states on the boundary. };
 	\item $\cL_{M_3}$: determined by the topological Chern-Simons theory (representing the dynamics of the bulk).
 \end{itemize}
 
{\bf  3. Geometric Proof of Transversality:}
The problem is thus reduced to proving that the tangent spaces of these two submanifolds intersect trivially at the critical point. This is the content of our second major theorem:

\begin{theorem}\label{thm:geometric}
For geometric boundary data $m$ (see Definition \ref{df:boundary-data}), the intersection of the tangent spaces of the boundary real Lagrangian part and the bulk real Lagrangian part is trivial:
\begin{equation}
T_x{\mathcal L}_{\coh}(m)\cap T_x{\mathcal L}_{M_3}=\{0\}.
\end{equation}
where $x\in {\mathcal L}_{\coh}(m)\cap {\mathcal L}_{M_3}$ is the point in the phase space corresponding to the geometric 4-simplex. 
\end{theorem}

This geometric part of the proof is intricate. In the flat case \cite{Kaminski:2019dld}, geometry could be encoded using bi-vectors. However, in the presence of a nonzero cosmological constant, bi-vectors alone cannot consistently describe the geometry, and we must instead rely fundamentally on holonomies \cite{Haggard:2015ima}. By analyzing the variations of holonomies around the faces of a constantly curved 4-simplex, we demonstrate that the only tangent vector satisfying the linearization of both the closure constraints (relevant to $\cL_{\rm coh}(m)$) and the flatness constraints (relevant to $\cL_{M_3}$) is the zero vector.

Both $\cL_{\rm coh}(m)$ and $\cL_{M_3}$ can be defined purely geometrically in terms of the phase space of Chern-Simons theory. Thus, the statement and the proof of the Theorem \ref{thm:geometric} do not rely on the particular construction of the model. This is the result which we expect to have a more general scope and to be applicable in various spin foam models based on Chern-Simons theory. For brevity, we first prove Theorem \ref{thm:geometric} in Section \ref{sec:classical} and only after that consider a particular model in Section \ref{sec:critical_spinfoam}.

\section{Stationary phase}
\label{sec:stationary_phase}

In this section, we introduce certain objects (real Lagrangian part) which play a crucial role in semiclassical analysis of spinfoam models. We will describe stationary point analysis in symplectic geometry terms and relate non-degeneracy of the Hessian to some property of the intersection of real Lagrangian parts of the actions of semiclassical states. We start with the simplest situation of states in $L^2(\R^N)$, but extend it to more complicated situations later.

\subsection{Stationary phase on $\R^N$}
\label{sec;stationary-RN}

In the simplest situation, we have two states (or generalized states) on $\R^N$ parametrized by $k\in \ZZ_+$ and given by 
\begin{equation}
\psi_k^\pm(\vec{q})=A^\pm(k,\vec{q})e^{ikS^\pm(\vec{q})}\,,
\end{equation}
where $\Im S^\pm\geq 0$ and $A^\pm(k,\vec{q})$, $S^\pm(\vec{q})$ are smooth functions on $\R^N$.  Throughout this paper, we only consider actions $S$ that satisfy $\Im S\geq 0$. We denote 
\begin{equation}
\vec{q}=(q^1,q^2,\ldots ,q^{N}).
\end{equation}
Typically, $A^\pm$ admits expansions in powers of $k$ (it is in the so-called symbol class \cite{Grigis_Sjoestrand_1994}).

We are interested in studying the asymptotic regime of a scalar product of two states, where $k\rightarrow \infty$.
\begin{equation}
\langle \psi^+_k,\psi^-_k\rangle_{\R^N}=\int_{\R^{N}}\rd^N\vec{q}\ \overline{\psi^+_k(\vec{q})}\psi^-_k(\vec{q}), \quad \rd^N\vec{q}:=\prod_{i=1}^{N}\rd q^i \,.
\end{equation}
The standard method is the stationary point analysis.
We denote
\begin{equation}
S_{tot}(\vec{q})=S_--\overline{S_+}\,.
\end{equation}
\begin{enumerate}
\item A real stationary point $\vec{q}_\ast$ is defined by conditions that
\begin{equation}
\left.\frac{\partial S_{tot}}{\partial q^i}\right|_{\vec{q}_\ast}=0,\ i=1,\ldots, N,\quad \Im S_{tot}(\vec{q}_\ast)=0\,.
\end{equation}
We will denote the set of such points by $\Stat(S_{tot})\subset \R^N$.
\item Hessian at the stationary point $\vec{q}_\ast$ is the matrix of second derivatives
\begin{equation}
{\bf H}(S_{tot})_{\vec{q}_\ast}=\partial^2 S_{tot}|_{\vec{q}_\ast}\,.
\end{equation}
The important condition for the stationary phase approximation is that ${\bf H}_{\vec{q}_\ast}$ is non-degenerate ($\det {\bf H}_{\vec{q}_\ast}\not=0$). If this condition is satisfied, the set of stationary points is discrete.
\end{enumerate}
One can show by the integration by parts technique that, if $A^\pm$ are in the symbol class $S^m_{\rho,\delta}$ (with $k$ treated as momentum)\footnote{
\blue{Let $X\subset \R^N$ be an open set, $0< \rho\leq 1, 0\leq \delta< 1, m\in\R, M\in\N\backslash\{0\}$. The symbol class $S^m_{\rho,\delta}(X,\R^M)$  of order $m$ and of type $(\rho,\delta)$ is the space of all $A\in C^\infty (X\times \R^M)$ such that for any compact $K\subset X$ and for all $\alpha\in\N^N,\beta\in \N^M$, there is a constant $C_{\alpha,\beta}$ such that 
\be
\left|\partial_x^\alpha \partial_{\xi}^\beta A(x, \xi)\right| \leq C_{\alpha \beta}(1+|\xi|)^{m-\rho|\beta|+\delta|\alpha|}\,.
\nn\ee
It guarantees sufficient decay of derivatives in $x$ (controlled by $\delta$) to justify the stationary phase approximation using integration by parts. See \cite{Grigis_Sjoestrand_1994} for more details. Here, $M=1$ and $\xi=k$.}}
 and they are functions of compact support in variables $\vec{q}$, then the stationary phase approximation gives the correct asymptotic expansion. However, in the situations considered in this paper, the assumption of compact support needs to be relaxed; thus, the applicability of the stationary phase method is an open question. We only mention that the absolute convergence of the integrals is not sufficient. \blue{The second generalization that we will need is 
to consider actions that admit singularities. We will assume
that $S_\pm$ are smooth and only defined on open subsets $X_\pm\subset \R^N$. The total action $S_{tot}$ is then defined on $X_+\cap X_-$ and $\Stat(S_{tot})\subset X_+\cap X_-$.}

We would like to express the stationary point set in terms of intersections of some \blue{subsets of the phase space} associated to $S_+$ and $S_-$. 
In the case where actions are real, the corresponding submanifolds are Lagrangian submanifolds of the phase space $T^\ast \R^N$ with symplectic form $\Omega$,
\begin{equation}\label{eq:Omega-TR}
\Omega=\frac{1}{2\pi}\sum_{i=1}^N \rd p_i\wedge \rd q^i\,.
\end{equation}
The situation is more complicated if the action is complex. The proper concept of positive Lagrangians was introduced in \cite{Mellin-Sjostrand-75}. It captures the whole asymptotic expansion, but for us, what is important will be only a fraction of the information encoded in these objects. This piece of information is given by the real Lagrangian part.

Let us introduce the notation for $T^\ast\R^N$: we list first momenta and then positions in the same order. We will also use the following notation for momenta and positions,
\begin{equation}
\form{p}=(p_1,p_2,\ldots ,p_{N}),\quad \vec{q}=(q^1,\ldots, q^N)\,.
\end{equation}
We also introduce the projection into positions $\pi\colon T^\ast\R^N\rightarrow \R^N$.

\begin{definition}\label{df:real-Lagriangian-part}
An action on \blue{an open subset $X\subset \R^N$} is a complex smooth function $S$ \blue{on $X$} satisfying $\Im S\geq 0$.
The {\bf real Lagrangian part} ${\mathcal L}^r_S$ for the action $S$ is a subset of $T^\ast\R^N$,
\begin{equation}
{\mathcal L}^r_S=\left\{(\form{p},\vec{q})\in T^\ast\R^N\colon \blue{\vec{q}\in X,\ }\frac{p_i}{2\pi}=\frac{\partial \Re S}{\partial q^i},\ i=1,\ldots, N,\quad \Im S=0\right\}.
\label{eq:def-Lr}
\end{equation}
\end{definition}

We can justify the introduction of our definition by the following fact:

\begin{prop}\label{prop:stat-0}
Let $S_\pm$ be actions on $\blue{X_\pm\subset} \R^N$. Denote $S_{tot}=S_--\overline{S_+}$ \blue{(defined on $X_+\cap X_-$)}. The map $\pi\colon T^\ast\R^N\rightarrow \R^N$ provides a bijection
\begin{equation}
\pi \colon {\mathcal L}^r_{S_+}\cap {\mathcal L}^r_{S_-}\rightarrow \Stat(S_{tot})\,.
\end{equation}
\end{prop}

\begin{proof}
Let us notice that
\begin{equation}
\Im S_{tot}=\Im \left(S_--\overline{S_+}\right)=\Im S_-+\Im S_+\geq 0\,.
\end{equation}
Moreover, the equality holds if and only if $\Im S_\pm=0$.  Additionally, at such points, $\rd \Im S_\pm=0$. 

Taking this into account, the condition for a stationary point is equivalent to
\begin{equation}
\frac{\partial \Re S_-}{\partial q^i}-\frac{\partial \Re S_+}{\partial q^i}=0\,,\quad \forall\, i=1,\cdots,N
\end{equation}
together with $\Im S_\pm=0$. Introducing 
\begin{equation}
p_i=2\pi\frac{\partial \Re S_\pm}{\partial q^i}\,,
\end{equation}
we associate with every element of $\Stat(S_{tot})$ an element of ${\mathcal L}^r_{S_+}\cap {\mathcal L}^r_{S_-}$, proving the bijection.
\end{proof}

The goal of this section is to analyze the condition for non-degeneracy of the Hessian in terms of real Lagrangian parts of corresponding actions. In general, these objects may have many singularities. We need to impose certain regularity conditions. 

\begin{definition}\label{df:real-Lagriangian-smooth}
Let $S$ be an action on $\blue{X\subset \R^N}$.
A point $x\in {\mathcal L}^r_S$ is {\bf regular} for action $S$ if the matrix of second derivatives $\partial^2\Im S$ has constant rank in an open neighborhood of $\pi(x)\in \R^N$. We denote the set of regular points by ${\mathcal L}^{\ro}_S$. It is an open and dense subset of ${\mathcal L}^r_S$.
\end{definition}

The condition ensures that around $\vec{q}_*\in {\mathcal L}^{\ro}_S$, the set $\{\vec{q}\colon \Im S=0\}$ is a submanifold, thus ${\mathcal L}^r_S$ is a smooth submanifold around regular points. We can say even more:

\begin{lemma}\label{lm:easy-S-0}
Let $S$ be an action on $\blue{X\subset \R^N}$ and $x\in T^\ast\R^N$ a regular point for $S$. Then, there exists an open neighborhood $U$ of $\pi(x)$ such that 
\begin{equation}
\{\vec{q}\in U\colon \Im S=0\}=\{\vec{q}\in U\colon \rd \Im S=0\}\,.
\end{equation}
Moreover, $\{\vec{q}\in U\colon \Im S=0\}$ is a manifold.
\end{lemma}

\begin{proof}
As $\Im S\geq 0$, the set $\{\vec{q}\in U\colon\Im S=0\}$ consists of critical points of $\Im S$ thus
\begin{equation}
\{\vec{q}\in \blue{X}\colon \Im S=0\}\subset \{\vec{q}\in \blue{X}\colon \rd \Im S=0\}\,.
\end{equation}
From regularity, there exists an open neighborhood $U'$ of $\pi(x)$ such that $\partial^2\Im S$ has constant rank. By the constant rank theorem, this means that 
\begin{equation}
B_{U'}=\{\vec{q}\in U'\colon \rd \Im S=0\}
\end{equation}
is a manifold. Taking a smaller neighborhood $U$ of $\pi(x)$, we can assume that $B_U$ is connected. On $B_U$ manifold, $\rd \Im S=0$ thus $\Im S$ is constant. However, $\Im S|_{\pi(x)}=0$ so 
\begin{equation}
\vec{q}\in B_U\Longrightarrow\left.\Im S\right|_{\vec{q}}=0\,,
\end{equation}
thus $B_U\subset \{\vec{q}\in U:\Im S=0\}$.
\end{proof}

We can state our main tool in its simplest version:

\begin{prop}\label{prop:clean-0}
Let $S_\pm$ be actions on $\blue{X\subset \R^N}$ and denote $S_{tot}=S_--\overline{S_+}$ \blue{(defined on $X_+\cap X_-$)}. Consider a point $x\in {\mathcal L}^{\ro}_{S_+}\cap {\mathcal L}^{\ro}_{S_-}$. Then
\begin{equation}
\det {\bf H}(S_{tot})_{\pi(x)}\not=0\Longleftrightarrow T_x{\mathcal L}^{\ro}_{S_+}\cap T_x{\mathcal L}^{\ro}_{S_-}=\{0\}.
\end{equation}
where $T{\mathcal L}^{\ro}_{S_\pm}$ are tangent spaces of ${\mathcal L}^{\ro}_{S_\pm}$ respectively as submanifolds of $T^\ast\R^N$.
\end{prop}

For brevity of exposition, we first prove a simple lemma.

\begin{lemma}\label{lm:easy-S-1}
Let $R, M_\pm$ be three real symmetric $n\times n$ matrices such that $M_\pm\geq 0$, namely,
\begin{equation}
\forall v\in \R^n,\quad v^T M_\pm v\geq 0\,.
\end{equation}
Then, the following conditions are equivalent for $v^\C\in \C^n$:
\begin{enumerate}
\item \label{it:1} $(R+iM_++iM_-)v^\C=0$.
\item \label{it:2}$Rv^\C=0$ and $M_\pm v^\C=0$.
\end{enumerate}
\end{lemma}

\begin{proof}
The only non-trivial direction is $\ref{it:1}\Longrightarrow \ref{it:2}$. Suppose, $(R+iM_++iM_-)v^\C=0$ then as, $R$, $M_\pm$ are Hermitian,
\begin{equation}
\Im \left(\overline{v^\C}^T(R+iM_++iM_-)v^\C\right)=\overline{v^\C}^TM_+v^\C +\overline{v^\C}^TM_-v^\C=0\,.
\end{equation}
But from positivity, $M_\pm v^\C=0$ and, as a consequence, $Rv^\C=0$.
\end{proof}

\begin{proof}[Proof of Proposition \ref{prop:clean-0}]
Suppose that the Hessian $\partial^2S_{tot}$ is degenerate, then there exists a complex vector $\vec{u}=\sum_i u^i\frac{\partial}{\partial q^i}$ such that $\partial^2S_{tot}\vec{u}=0$. It can be written in the form
\begin{equation}
(\partial^2 \Re S_{tot}+i\partial^2 \Im S_++i\partial^2 \Im S_-)\vec{u}=0\,,
\end{equation}
where $\partial^2 \Re S_{tot}=\partial^2 \Re S_{-}-\partial^2 \Re S_{+}$.
Let us notice that $\Im S_\pm\geq 0$ and that, at the stationary point, $\Im S_\pm=0$, thus it is a local minimum for both functions ($\partial^2\Im S_\pm\geq 0$). Applying Lemma \ref{lm:easy-S-1}, we can show that the condition for $\vec{u}$ being in the kernel of Hessian is equivalent to
\begin{equation}
\partial^2 \Im S_\pm \vec{u}=0\,,\quad (\partial^2\Re S_--\partial^2\Re S_+)\vec{u}=0\,.
\label{eq:kernal_condition}
\end{equation}
Consider now a vector $v$ in $T_x(T^\ast \R^N)$ parametrized as
\begin{equation}
v=(\form{w},\vec{u}),\quad \form{w}=\sum_i w_i\frac{\partial}{\partial p_i},\quad \vec{u}=\sum_i u^i\frac{\partial}{\partial q^i}\,.
\end{equation}
We will find the conditions for $v\in T_x{\mathcal L}^{\ro}_{S_+}\cap T_x{\mathcal L}^{\ro}_{S_-}$.

From the assumption of regularity and Lemma \ref{lm:easy-S-0}, locally around the point $x$,
\begin{equation}
{\mathcal L}^{\ro}_{S_\pm}=\left\{\frac{p_i}{2\pi}=\frac{\partial \Re S_\pm}{\partial q^i},\ \frac{\partial \Im S_\pm}{\partial q^i}=0\right\}\,.
\end{equation}
The condition for $v\in T_x{\mathcal L}^{\ro}_{S_\pm}$ is given by annihilation of the equations determining these manifolds (regularity plays the role here). The conditions are given by
\begin{equation}
\f{\form{w}}{2\pi}=\partial^2\Re S_\pm \vec{u},\quad \partial^2\Im S_\pm \vec{u}=0\,.
\end{equation}
Eliminating $\form{w}$, the conditions for $\vec{u}$ reduce to the conditions for the kernel of the Hessian \eqref{eq:kernal_condition}. This shows equivalence.
\end{proof}

We now provide some basic examples.
\begin{example}
Let $S$ be an action on $\blue{X\subset \R^N}$ that is real. Then,
\begin{equation}
{\mathcal L}^r_S={\mathcal L}^{\ro}_S=\left\{(\form{p},\vec{q})\colon \frac{p_i}{2\pi}=\frac{\partial S}{\partial q^i}\blue{,\ \vec{q}\in X}\right\}\,.
\end{equation}
It is known that ${\mathcal L}^r_S$ is a Lagrangian submanifold.
\end{example}

We can now state an important example of a real Lagrangian part associated to a complex action, which is another extreme in comparison to the real action example.

\begin{example}\label{ex:coherent}
Let $S$ be an action on $\blue{X\subset \R^N}$ and $\vec{q}_\ast\in\blue{X}$. Suppose that
\begin{equation}
\{\vec{q}\in\blue{X}\colon\left.\Im S\right|_{\vec{q}}=0\}=\{\vec{q}_*\}\,,
\end{equation}
and that the Hessian of $\Im S$ at $\vec{q}_\ast$ is strictly positive. Let
\begin{equation}
\frac{p^\ast_i}{2\pi}=\left.\frac{\partial \Re S}{\partial q^i}\right|_{\vec{q}_\ast}.
\end{equation}
Then
\begin{equation}
{\mathcal L}^r_S={\mathcal L}^{\ro}_S=\{(\form{p}_\ast,\vec{q}_\ast)\}.
\end{equation}
We will call such $S$ a {\bf coherent state action} peaked at $(\form{p}_\ast,\vec{q}_\ast)$.
\end{example}

Indeed, in this case the matrix $\partial^2\Im S$ has maximal rank $N$ at $\vec{q}_\ast$. Thus, it has the same rank in some open neighborhood of $\vec{q}_\ast$.
The description of the action matches the semiclassical definition of a coherent state.

 Let us finish this subsection with a few observations. Firstly, if $S$ is an action, then $-\overline{S}$ is an action as well.  We introduce a map
\begin{equation}
I\colon T^\ast\R^N\rightarrow T^\ast\R^N,\quad I(\form{p},\vec{q})=(-\form{p},\vec{q})\,.
\end{equation}
Using $I$, we can describe the real Lagrangian part for $-\overline{S}$ as follows:
\begin{equation}
{\mathcal L}^r_{-\overline{S}}=I\left({\mathcal L}^r_S\right),\quad {\mathcal L}^{\ro}_{-\overline{S}}=I\left({\mathcal L}^{\ro}_S\right)\,.
\end{equation}
The second observation concerns the sum of actions. Let $S(\vec{q}_+,\vec{q}_-)=S_+(\vec{q}_+)+S_-(\vec{q}_-)$. Using identification $T^\ast \R^{2N}=T^\ast \R^{N}\times T^\ast \R^{N}$, we can write
\begin{equation}
{\mathcal L}^r_S={\mathcal L}^r_{S_+}\times {\mathcal L}^r_{S_-},\quad {\mathcal L}^{\ro}_S={\mathcal L}^{\ro}_{S_+}\times {\mathcal L}^{\ro}_{S_-}\,.
\end{equation}
We will use this property extensively.

\subsection{Integral kernels}
\label{sec:integral kernel}

In the semiclassical limit, the states are described by real Lagrangian parts. In the situation considered in this paper, this is achieved through Proposition \ref{prop:clean-0}. The semiclassical description of operators is expected to be in terms of canonical relations  \cite{guillemin2013semi, hormander-IV}. In particular, unitary operators should be described by symplectic transformations. We will now implement this general rule in a specific situation. Under certain conditions on the family of unitary operators $U_k$, we will associate with it a symplectic transformation $\chi_U\colon T^\ast \R^N\rightarrow T^\ast \R^N$. 
We can write $\langle \psi_k^+,U_k\psi_k^-\rangle_{\R^N}$ for any two semiclassical states $\psi^\pm_k$ as an oscillatory integral to which the stationary phase analysis can be applied. 
The symplectic transformation $\chi_U$ will appear in Proposition \ref{lm:U-asym} which describes stationary points of the action and conditions for non-degeneracy of the Hessian.

Let us start by considering a family of operators $U_k\colon L^2(\R^N)\rightarrow L^2(\R^N)$ labeled by an integer $k$. Suppose that their integral kernels are given by
\begin{equation}\label{eq:U_k}
U_k(\vec{q}_+,\vec{q}_-)=\int_{\R^{N_o}}\rd^{N_o}\vec{q}_o\ A(k,\vec{q}_+,\vec{q}_-,\vec{q}_o)e^{ikS_U(\vec{q}_+,\vec{q}_-,\vec{q}_o)}\,.
\end{equation}
For two asymptotic states $\psi^\pm(\vec{q}^\pm)=A_\pm e^{ikS_\pm}$ on $\R^{N}$, we are interested in the asymptotic expansion of $\langle \psi^+,U_k\psi^-\rangle_{\R^N}$.
This can be written in the form of an oscillatory integral: 
\begin{equation}
\langle \psi^+,U_k\psi^-\rangle_{\R^N}=\int_{\R^{2N+N_o}}\prod_{i=1}^{N}\rd q_+^i\prod_{i=1}^{N}\rd q_-^i\prod_{i=1}^{N_o}\rd q_o^i\ 
\overline{A_+(k,\vec{q}_+)}A_-(k,\vec{q}_-)A(k,\vec{q}_+,\vec{q}_-,\vec{q}_o)e^{ik S_{tot}}\,.
\end{equation}
\blue{Our goal is to study stationary points of this integral, restricting to the region where the integrated function is non-singular. We assume now that $S_\pm$ are actions on $X_\pm\subset \R^N$. The relevant action for the analysis} is $S_{tot}\colon \blue{X_+\times X_-}\times \R^{N_o}\rightarrow \C$,
\begin{equation}\label{S-tot-UU}
S_{tot}(\vec{q}_+,\vec{q}_-,\vec{q}_o)=-\overline{S}_+(\vec{q}_+)+S_-(\vec{q}_-)+S_U(\vec{q}_+,\vec{q}_-,\vec{q}_o)\,.
\end{equation}
We will now describe the stationary points for this action and the conditions for non-degeneracy of the Hessian under some assumptions about $S_U$. 

We introduce a submanifold of $T^\ast\R^{2N+N_o}$: 
\begin{equation}
    P_o:=\{(\form{p}_+,\vec{q}_+,\form{p}_-,\vec{q}_-,\form{p}_o,\vec{q}_o)\in T^\ast \R^{N}\times T^\ast \R^{N}\times T^\ast \R^{N_o}\colon \form{p}_o=0\}.
\end{equation}

\begin{definition}
Let $S\colon \R^{2N+N_o}\rightarrow \R$ be a real action. We call it a {\bf generating function of symplectic transformation} if there exists a diffeomorphism $\chi\colon T^\ast\R^{N}\rightarrow T^\ast\R^{N}$ and a map $ \vec{\xi}\colon T^\ast\R^{N}\rightarrow \R^{N_o}$ such that 
\begin{enumerate}
    \item The subspace $M_{S}={\mathcal L}^r_S\cap P_o\subset T^\ast \R^{N}\times T^\ast \R^{N}\times T^\ast \R^{N_o}$ can be expressed by
    \begin{equation}
 M_{S}=\left\{(\form{p}_+,\vec{q}_+,\form{p}_-,\vec{q}_-,0,\vec{q}_o)\colon (\form{p}_+,\vec{q}_+)=\chi\left(I(\form{p}_-,\vec{q}_-)\right),\ \vec{q}_o=\xi(\form{p}_-,\vec{q}_-)\right\}\,.
 \label{eq:M_S}
    \end{equation}
    \item $M_S$ is a clean intersection of $\cL_S^r$ and $P_0$. That is, for every point $x\in  M_{S}$,
    \begin{equation}
        T_x M_{S}=T_x{\mathcal L}^r_S\cap T_xP_o\,.
    \end{equation}
\end{enumerate}
We call $\chi$ a symplectic transformation generated by $S$.
\end{definition}

\blue{The variables $\vec{q}_o$ play an auxiliary role, and they will be called the auxiliary variables. The generating function depends on the initial variables $\vec{q}_-$, the final variables $\vec{q}_+$ and the auxiliary variables $\vec{q}_o$.}

 Such actions satisfy the {\it transverse generating function} condition of \cite{guillemin2013semi} (Chapter 5.1 and 5.2). In particular, $\chi$ is indeed a symplectic transformation, which justifies our notation. 

Additionally, if two actions $S_\pm\colon \R^{2N+N_o^\pm}\rightarrow \R$ are generating functions for symplectic transformations $\chi_\pm$ respectively, then an action
\begin{equation}
    S(\vec{q}_+,\vec{q}_-,\vec{Q}_0)=S_+(\vec{q}_+,\vec{q},\vec{q}_o^+)+S_-(\vec{q},\vec{q}_-,\vec{q}_o^-),\quad \vec{Q}_0=(\vec{q},\vec{q}_o^+,\vec{q}_o^-)
\end{equation}
is a generating function for \blue{the composition symplectic transformation} $\chi_+\circ \chi_-$. \blue{Here $\vec{q}_o^\pm$ are the auxiliary variables for actions $S_\pm$ respectively and $\vec{Q}_0$ are the auxiliary variables for  action $S$.}
Thus, the map that associates $\chi$ with an operator family $U_k$ is a morphism (it preserves the composition). For the theory of generating functions, we refer the reader to  Chapter 5 of \cite{guillemin2013semi}.

\begin{prop}\label{lm:U-asym}
 Let $S_U$ be a generating function of a symplectic transformation $\chi_U$, then the following holds:
\begin{enumerate}
\item Stationary points for an action $S_{tot}$ \eqref{S-tot-UU} are in bijection with
\begin{equation}
{\mathcal L}^r_{S_+}\cap \chi_U({\mathcal L}^r_{S_-})\,.
\end{equation}
\item For a point corresponding to $x\in {\mathcal L}^{\ro}_{S_+}\cap \chi_U({\mathcal L}^{\ro}_{S_-})$, the Hessian is non-degenerate if and only if
\begin{equation}
T_x{\mathcal L}^{\ro}_{S_+}\cap T_x\left(\chi_U({\mathcal L}^{\ro}_{S_-})\right)=\{0\}.
\end{equation}
\end{enumerate}
\end{prop}

\begin{proof}
We apply Propositions \ref{prop:stat-0} and \ref{prop:clean-0} to the action \blue{$\tilde{S}_+$ on $X_+\times X_-\times \R^{N_o}$ and the action $\tilde{S}_-$ on $\R^{N}\times\R^{N}\times \R^{N_o}$} defined by
\begin{align}
&\tilde{S}_+(\vec{q}_+,\vec{q}_-,\vec{q}_o)=S_+(\vec{q}_+)-\overline{S}_-(\vec{q}_-)\,,\\
&\tilde{S}_-(\vec{q}_+,\vec{q}_-,\vec{q}_o)=S_U(\vec{q}_+,\vec{q}_-,\vec{q}_o)\,.
\end{align}
The stationary points are in one-to-one correspondence with the intersection
\begin{equation}
{\mathcal L}^r_{\tilde{S}_+}\cap {\mathcal L}^r_{\tilde{S}_-}\subset T^\ast\R^N\times T^\ast\R^N\times T^\ast\R^{N_o}\,.
\label{eq:intersection_of_Lagra}
\end{equation}
Let us notice that, as $\tilde{S}_+$ does not depend on $\vec{q}_o$ and it is a sum of actions depending on separate sets of variables,
\begin{equation}
{\mathcal L}^r_{\tilde{S}_+}={\mathcal L}^r_{S_+}\times I\left({\mathcal L}^r_{S_-}\right)\times \{\form{p}_o=0\}\,,
\end{equation}
where $\{\form{p}_o=0\}$ denotes the Lagrangian $\{(\vec{q}_o,\form{p}_o)\in T^\ast\R^{N_o} : \form{p}_o=0 \}$.
Then the intersection \eqref{eq:intersection_of_Lagra} can be written as 
\begin{equation}
{\mathcal L}^r_{\tilde{S}_+}\cap {\mathcal L}^r_{\tilde{S}_-}= \left({\mathcal L}^r_{S_+}\times I\left({\mathcal L}^r_{S_-}\right)\times  \{\form{p}_o=0\}\right)\cap {\mathcal L}^r_{S_U}=
\left({\mathcal L}^r_{S_+}\times I\left({\mathcal L}^r_{S_-}\right)\times T^\ast \R^{N_o}\right)\cap \left({\mathcal L}^r_{S_U}\cap P_o\right)\,.    
\end{equation}
Using the condition for $S_U$, $x=(x_+,x_-,x_o)\in {\mathcal L}^r_{\tilde{S}_+}\cap {\mathcal L}^r_{\tilde{S}_-}$ if and only if
\begin{equation}
    x_+=\chi_U(I(x_-)),\quad x_o= \tilde{\xi}_U(x_-),\quad x_+\in {\mathcal L}^r_{S_+},\quad x_-\in I\left({\mathcal L}^r_{S_-}\right)\,,
\end{equation}
 where the map $\tilde{\xi}_U\colon T^\ast \R^N\rightarrow T^\ast \R^{N_o}$ is related to $\vec{\xi}_U$ introduced in the definition of a generating function by
\begin{equation}
    \tilde{\xi}_U(x)=(\form{0},\vec{\xi}_U(x))\,.
\end{equation}
Eliminating $x_-=I(\chi_U^{-1}(x_+))$ and $x_o= \tilde{\xi}_U(x_-)$ (notice that $I^2=\operatorname{id}$), we get
\begin{equation}
    {\mathcal L}^r_{\tilde{S}_+}\cap {\mathcal L}^r_{\tilde{S}_-}=\left\{(x_+,x_-,x_o)\colon x_-=I(\chi_U^{-1}(x_+)),\ x_o= \tilde{\xi}_U(x_-),\ x_+\in {\mathcal L}^r_{S_+}\cap \chi_U({\mathcal L}^r_{S_-})\right\}\,.
\end{equation}
we, therefore, obtain the first statement of the proposition.

For the Hessian, let us consider a vector $(v_+,v_-,v_o)$ in the intersection of the tangent spaces. We notice that (condition to be in the tangent space of ${\mathcal L}^{\ro}_{\tilde{S}_+}$)
\begin{equation}\label{eq:v-in-L}
v_+\in T_{x_+}{\mathcal L}^{\ro}_{S_+},\ v_-\in T_{x_-}\left(I{\mathcal L}^{\ro}_{S_-}\right) ,\ v_o\in T_{x_o}\{\form{p}_o=0\}\,.
\end{equation}
In particular, 
\begin{equation}
    (v_+,v_-,v_o)\in T_xP_o,
\end{equation}
and from the two properties of generating functions of the canonical relations
\begin{equation}\label{eq:S_U-v}
  (v_+,v_-,v_o)\in T_xM_{S_U}=T_xP_o\cap T_x{\mathcal L}^{\ro}_{S_U}\Longleftrightarrow v_+= D_{x_-}(\chi_UI)(v_-),\ v_o=D_{x_o} (\tilde{\xi}_U)(v_-)\,,
\end{equation}
where the map  $D_x(f):T_xX\rightarrow T_{f(x)}Y$ is the derivative of a map $f:X\rightarrow Y$.
We will now look for the intersection of both tangent spaces. The conditions are \eqref{eq:v-in-L} together with \eqref{eq:S_U-v}. After eliminating $v_-$ and $v_o$ by the relations
\begin{equation}
    v_-= \left[D_{x_-}(\chi_UI)\right]^{-1}(v_+)\,,\quad v_o=D_{x_o}(\tilde{\xi}_U)(v_-)\,.
\end{equation}
The conditions for $v_+$ reduce to
\begin{equation}
    v_+\in T_{x+}{\mathcal L}^{\ro}_{S_+},\quad 
     v_+\in D_{x_-}(\chi_UI)\left( T_{x_-}I{\mathcal L}^{\ro}_{S_-}\right)
    =T_{x_+}\left( \chi_U{\mathcal L}^{\ro}_{S_-}\right).
\end{equation}
In the last equality, we have used the fact that both $I$ and $\chi_U$ are diffeomorphisms and $I^2=\operatorname{id}$. This finishes the second part of the proposition.
\end{proof}
This proposition generalizes Proposition \ref{prop:clean-0}.

\subsection{Partial Hessians}
\label{sec:partial-Hessian}

We will now describe another important fact about Hessians. Consider an action $S(\vec{q},\vec{Q})$ on $\blue{X\subset\R^{n}\times \R^N}$. Let us denote $S_{red}(\vec{q})=S(\vec{q},\vec{Q}_*)$ for a fixed vector $\vec{Q}_*\in \R^N$. Our goal is to describe the stationary points $\Stat(S_{red})$ of $S_{red}$ and its Hessian ${\bf H}(S_{red})$ in terms of the full action.

Let us introduce another action on $\blue{X\times R^N\subset\R^{n}\times \R^N\times \R^N}$,
\begin{equation}
S_o(\vec{q},\vec{Q},\form{\lambda})=S(\vec{q},\vec{Q})+\sum_{i=1}^N\lambda_i(Q^i-Q^i_*)\,.
\label{eq:S_o}
\end{equation}
The reasoning behind this action is both the method of Lagrange multipliers and the integral formula for delta functions. The relation between the stationary points of both actions can be shown directly:

\begin{lemma}\label{lm:partial-Hessian}
Let $S$ be an action on $\blue{X\subset\R^{n}\times \R^N}$  and $S_o$, $S_{red}$ as described above. Then 
\begin{enumerate}
\item $(\vec{q},\vec{Q},\form{\lambda})\in \Stat(S_o)$ if and only if
\begin{equation}
\vec{Q}=\vec{Q}_\ast,\quad \lambda_i=-\frac{\partial S_o}{\partial Q^i},\ i=1,\ldots, N,\quad \vec{q}\in \Stat(S_{red})\,.
\end{equation}
\item The Hessian ${\bf H}(S_o)$ satisfies
\begin{equation}
\det {\bf H}(S_o)=(-1)^N\det {\bf H}(S_{red})
\label{eq:detS_o}
\end{equation}
at the corresponding stationary points.
\end{enumerate}
\end{lemma}

\begin{proof}
The first part of the lemma is the Lagrange multiplier method. The second part is a direct computation.
\end{proof}

In particular, Lemma \ref{lm:partial-Hessian} allows us to extend the results of Proposition \ref{lm:U-asym} to the case when the integral kernel involves delta functions on part of the variables. Such an extension is straightforward.

\subsection{Metaplectic group}

We will now focus on the case when the integral kernel emerges from \blue{a} metaplectic transformation. Our goal is to give a realization of $\chi_U$ defined in Section \ref{sec:integral kernel}, leading to a corresponding version of Proposition \ref{lm:U-asym} ({\it r.f.} Proposition \ref{prop:clean-o}).

In quantum theory, the space of affine operators is very important. For $v=(\form{w},\vec{u})$ and $a\in \R$ we consider a symmetric operator,
\begin{equation}
\hat{H}=a+\sum_i w_i\hat{q}^i-u^i\hat{p}_i,\quad \hat{p}_i=-\frac{2\pi i}{k}\frac{\partial}{\partial q^i},
\end{equation}
where $\hat{q}^i$ denotes multiplication by $q^i$. In our convention,  the Planck constant is $\hbar=\f{2\pi}{k}$.

The Weyl operator $W_k(v,a):=e^{ik\hat{H}}$ is an unitary operator. Its action can be easily computed:
\begin{equation}
W_k(v,a)\phi(\vec{q})=e^{i\left(ka-\pi\sum_i w_i u^i\right)}e^{ik\sum_i w_iq^i}\phi\left(\vec{q}-2\pi\vec{u}\right)\,.
\end{equation}
The Weyl operators satisfy the relation
\begin{equation}
W_k(v,a)W_k(v',a')=W(v+v',a+a'-2\pi^2 \Omega(v,v')), \quad W_k(v,a)=e^{ika}W_k(v,0).
\end{equation}
In particular, the adjoint action has a form
\begin{equation}
W_k(v,a)W_k(v',a')W_k(v,a)^{-1}=W_k(v',a'-4\pi^2\Omega(v,v'))\,.
\end{equation}
Consider affine canonical (symplectic) transformations on $T^\ast\R^N$. They form a group $\Aff(2N,\R)$ that can be identified with
\begin{equation}
\Aff(2N,\R)=\Sp(2N,\R)\ltimes T^\ast\R^N\,,
\end{equation}
\blue{where $\Sp(2N,\R)$ is the symplectic group represented by $2N\times 2N$ matrices with real entries. }
For $(M,v)$ with $M\in \Sp(2N,\R)$ and $v=(\form{w},\vec{u})\in T^\ast\R^N$, the action on $T^\ast\R^N$ is given by
\begin{equation}
(M,v)\cdot(\form{p},\vec{q})=M(\form{p},\vec{q})+(\form{w},\vec{u})\,.
\end{equation}
A metaplectic implementer of $H=(M,v)\in \Aff(2N,\R)$ is a unitary operator $U_{(M,v),k}$ on $L^2(\R^N)$ with the special property that the adjoint action induces the expected automorphisms of the Weyl algebra:
\begin{equation}
U_{(M,v),k}W_k(v',a)U_{(M,v),k}^{-1}=W_k(Mv',a-4\pi^2\Omega(v,Mv'))\,.
\end{equation}
Such $U_{(M,v),k}$ is determined uniquely up to the phase. In particular,
\begin{equation}
U_{({\mathbb I},v),k}=W_k(v,a)\,,
\end{equation}
where $a$ is arbitrary (phase factor). Metaplectic implementers form a group denoted as $\Met_{N}$.
There is a group homomorphism $\met_k\colon\Met_N\rightarrow \Aff(2N,\R)$ defined by the property $\met_k(U)=(M,v)$, where $(M,v)$ satisfies
\begin{equation}
UW_k(v',a)U^{-1}=W_k(Mv',a-4\pi^2\Omega(v,Mv'))\,,
\end{equation}
with the kernel given by the group $U(1)$ of phases.

\subsection{Metaplectic implementers}
\label{subsec:metaplectic}

We will now present some nice properties of metaplectic implementers.
For a metaplectic implementer of $H\in \Aff(2N,\R)$, the integral kernel can be expressed by a Gaussian integral \blue{(see \cite{derezinski2013mathematics}, Chapter 10)}. 
\begin{equation}\label{eq:U-S-formula-0}
U_{H,k}^{a,S_H}(\vec{q}_+,\vec{q}_-)=C_ke^{ika}\int_{\R^{N_o}}d^{N_o}\vec{q}_o\ e^{ikS_H}\,,
\end{equation}
where $C_k$ is the normalization constant (uniquely determined positive constant, that is homogeneous in $k$ of some rational order), $a\in \R$ is a phase and $S_H(\vec{q}_+,\vec{q}_-,\vec{q}_o)$ is a real polynomial of degree at most two. 

Let us describe the actions for metaplectic transformations in some detail. First, we can change the variables $\vec{q}_o$  linearly, such that they are separated into two parts $\vec{q}_o'$ and $\vec{\lambda}$ and that $S_H$ depends quadratically on $\vec{q}_o'$ and linearly on $\vec{\lambda}$. One can perform a Gaussian integration of $\vec{q}_o'$. If the Hessian for the new action is non-degenerate, then it is also non-degenerate for the original action \blue{$S_H$}. Thus, we can always assume that $S_H$ is linear in $\vec{q}_o=\vec{\lambda}$.

Such a minimal version of $S_H$ can be found as follows. 
Consider an affine canonical transformation
\begin{equation}\label{eq:canonical}
\form{p}^+(\form{p}^-,\vec{q}_-),\ \vec{q}_+(\form{p}^-,\vec{q}_-)\,.
\end{equation}
If the functions $\vec{q}_+$ and $\vec{q}_-$ are independent, then they can be used as a coordinate system. In this situation, there exists a generating function for the canonical transformation.  In general, there might be some dependencies between these variables. There always exist independent affine functions $f_\alpha(\vec{q}_+,\vec{q}_-)$, $\alpha=1,\ldots, N_o$ (the set might be empty if $N_o=0$) such that
\begin{equation}
f_\alpha(\vec{q}_+(\form{p}^-,\vec{q}_-),\vec{q}_-)=0,\quad \alpha=1,\ldots, N_o\,.
\end{equation}
The canonical transformation is described by the generalized generating function $S_{H,o}(\vec{q}_+,\vec{q}_-)$ (polynomial of degree at most 2) that satisfies the identity:
\begin{equation}\label{eq:S_H_lambda}
\pm \frac{p^\pm_i}{2\pi}=\frac{\partial S_{H,o}}{\partial q_\pm^i}+\sum_{\alpha=1}^{N_o}\lambda^\alpha \frac{\partial f_\alpha}{\partial q_\pm^i},\ i=1,\ldots, N,\quad f_\alpha=0,\quad \alpha=1,\ldots, N_o\,.
\end{equation}
The action for the implementer, $S_H$ on $\R^N\times \R^N\times \R^{N_o}$, is 
\begin{equation}
S_H(\vec{q}_+,\vec{q}_-,\vec{\lambda})=S_{H,o}+\sum_{\alpha=1}^{N_o}\lambda^\alpha f_\alpha\,,
\end{equation}
where we have denoted $\vec{q}_o=\vec{\lambda}$. For later convenience, we notice that \eqref{eq:S_H_lambda} allows us to determine $\vec{\lambda}$ in terms of $\form{p}^-,\vec{q}_-$,
\begin{equation}
\vec{\lambda}=\vec{\Lambda}(\form{p}^-,\vec{q}_-)\,.
\end{equation}
 Note that $S_H$ satisfies the assumption of Section \ref{sec:integral kernel} with $\chi_U=H$,  $\vec{\xi}_U=\vec{\Lambda}$. Our discussion of reduction to minimal implementer shows that this is also true for any Gaussian action for the implementer.  The kernel of the composition of two implementers of $H_+$ and $H_-$  can be written again as a Gaussian oscillatory integral. This procedure allows us to extend the result on non-degeneracy of the Hessian to the case where $U_{H,k}$ is written as a product of basic implementing operations as in \cite{Han:2021tzw, Han:2025mkc}. In summary, applying Proposition \ref{lm:U-asym}, we obtain the following result:

\begin{prop}\label{prop:clean-o}
Let $S_\pm$ be two actions and $S_H$ be an action for the implementer of $H\in \Aff(2N,\R)$. Denote 
\begin{equation}
S_{tot}(\vec{q}_+,\vec{q}_-,\vec{q}_o)=-\overline{S_+}(\vec{q}_+)+S_-(\vec{q}_-)+S_H(\vec{q}_+,\vec{q}_-,\vec{q}_o)\,.
\end{equation}
Then,
\begin{enumerate}
\item Stationary points for an action $S_{tot}$  are in bijection with
\begin{equation}
{\mathcal L}^r_{S_+}\cap H({\mathcal L}^r_{S_-})\,.
\end{equation}
\item For a point corresponding to $x\in {\mathcal L}^{\ro}_{S_+}\cap H({\mathcal L}^{\ro}_{S_-})$, the Hessian is non-degenerate if and only if
\begin{equation}
T_x{\mathcal L}^{\ro}_{S_+}\cap T_x\left(H({\mathcal L}^{\ro}_{S_-})\right)=\{0\}\,.
\end{equation}
\end{enumerate}
\end{prop}

Additionally,  Lemma \ref{lm:partial-Hessian} allows us to extend the results of Proposition \ref{prop:clean-o} to the case when the integral kernel involves delta functions on part of the variables.

\subsection{\blue{Complex variables}}
\label{sec:complex-phase-space}

Let us generalize some of the results above to phase spaces with complex coordinates, which will be shown to be useful in Section \ref{sec:critical_spinfoam}. 
Consider $\FF=\C^{2N}$ with complex coordinates  $p^\C_i,q^i_\C$, $i=1,\ldots, N$ on the corresponding copies of $\C$ and a symplectic form $\Omega$ as follows. 
\begin{equation}\label{eq:symp-Darboux}
\Omega=\frac{1}{4\pi}\left(s\Omega^\C+\overline{s}\overline{\Omega^\C}\right),\quad \Omega^\C=\sum_i\rd p_i^\C\wedge \rd q_\C^i\, ,
\end{equation}
where $s\in \C$ is a fixed parameter.

We identify $\FF$ with the real phase space $T^*\R^{2N}$ by choosing a polarization in which the complex variables $\vec{q}_\C$ serve as the configuration coordinates. This identification, denoted by 
\be
\Iota \colon \C^{2N}\rightarrow T^*\R^{2N}\,,
\label{eq:Iota}
\ee 
is defined explicitly by:
\begin{equation}\label{eq:t-identification}
\vec{q}=(\Re q^1_\C,\ldots, \Re q^N_\C,\Im q^1_\C,\ldots, \Im q^N_\C),\quad \form{p}=\left(\Re \left(sp_1^\C\right),\ldots, \Re\left(sp_N^\C\right),-\Im\left( sp_1^\C\right),\ldots, -\Im\left(sp_N^\C\right)\right).
\end{equation}
Throughout this subsection, we implicitly use the natural identification $\C^N\ni\vec{q}_\C\mapsto \vec{q}\in \R^{2N}$ for the configuration space. We emphasize that the symplectic map $\Iota$ depends on the choice of the parameter $s$, which is held fixed.

\begin{lemma}\label{lm:real-Lagr-complex}
Let $S\colon \blue{X^\C}\rightarrow \C$ be an action \blue{for an open set $X^\C\subset \C^N$}, then
\begin{equation}
\Iota^{-1}\left({\mathcal L}_S^r\right)=\left\{(\form{p}^\C,\vec{q}_\C)\colon \blue{\vec{q}_\C\in X^\C,\ }\frac{s}{4\pi}p_i^\C=\frac{\partial \Re S}{\partial q^i_\C},\ i=1,\cdots,N,\, \Im S=0\right\}.
\end{equation}
\end{lemma}

Here, we have used the holomorphic derivative $\frac{\partial}{\partial q^i_\C}:=\frac{1}{2}\left(\frac{\partial}{\partial \Re q^i_\C}-i\frac{\partial}{\partial \Im q^i_\C}\right)$.

\begin{proof}
A simple computation using the definition of holomorphic derivative shows that 
\begin{align}
&2\Re\left(\frac{s}{4\pi}p_i^\C-\frac{\partial \Re S}{\partial q^i_\C}\right)=\frac{1}{2\pi}\Re\left(sp_i^\C\right)-\frac{\partial \Re S}{\partial \Re q^i_\C}=\frac{1}{2\pi}p_i-\frac{\partial \Re S}{\partial q^i}\,,\\
&2\Im\left(\frac{s}{4\pi}p_i^\C-\frac{\partial \Re S}{\partial q^i_\C}\right)=\frac{1}{2\pi}\Im\left(sp_i^\C\right)+\frac{\partial \Re S}{\partial \Im q^i_\C}=-\left(\frac{1}{2\pi}p_{i+N}-\frac{\partial \Re S}{\partial q^{i+N}}\right)\,.
\end{align}
The conditions stated in the lemma are equivalent to the conditions \eqref{eq:def-Lr} for ${\mathcal L}^r_S$.
\end{proof}

In a similar fashion, the symplectic transformations can be described in terms of the complex variables.

\begin{lemma}\label{lm:sympl-complex}
Let $S \colon \C^N\times \C^N\times \C^{N_o}\rightarrow \R$ be a generating function of a symplectic transformation $\chi_S$. Denote ${\mathcal H}_S=\Iota^{-1} \chi_S\Iota$. Then the following are equivalent for $(\form{p}^{\C\pm},\vec{q}_{\C\pm})\in \C^{2N}$. 
\begin{enumerate}
\item ${\mathcal H}_S(\form{p}^{\C-},\vec{q}_{\C-})=(\form{p}^{\C+},\vec{q}_{\C+})$\,,
\item There exists $\vec{q}_{\C o}$ such that $S(\vec{q}_{\C+},\vec{q}_{\C-},\vec{q}_{\C o})$ satisfies
\begin{equation}
\pm \frac{s}{4\pi}p_i^{\C\pm}=\frac{\partial S}{\partial q^i_{\C\pm}}\,,\, i=1,\cdots,N\quad \frac{\partial S}{\partial q^j_{\C o}}=0\,,\,j=1,\cdots,N_o\,.
\end{equation}
\end{enumerate}
\end{lemma}

\begin{proof}
Let us notice that
\begin{equation}
\frac{\partial S}{\partial q^j_{\C o}}=0\Longleftrightarrow \frac{\partial S}{\partial \Re q^j_{\C o}}=0\text{ and } \frac{\partial S}{\partial \Im q^j_{\C o}}=0\,.
\label{eq:project}
\end{equation}
Due to the properties of generating functions, this defines a map $\xi:T^\ast\R^{2N}\rightarrow \R^{2N_o}$. Further, $\pm \frac{s}{4\pi}p_i^{\C\pm}=\frac{\partial S}{\partial q^i_{\C\pm}}$ is equivalent to a set of real equations:
\begin{equation}
\pm \frac{1}{2\pi}p_i^{\pm}=\frac{\partial S}{\partial q^i_{\pm}},\quad \pm \frac{1}{2\pi}p_{i+N}^{\pm}=\frac{\partial S}{\partial q^{i+N}_{\pm}}
\label{eq:real_symp_tranf}
\end{equation}
as in the proof of Lemma \ref{lm:real-Lagr-complex}. 
Equations \eqref{eq:real_symp_tranf} are precisely the defining relations for the generating function of a symplectic transformation $\chi_S: T^\ast \R^{2N}\rightarrow T^\ast\R^{2N}$ on the real phase space such that $(\form{p}^+,\vec{q}_+)=\chi_S(\form{p}^-,\vec{q}_-)$, 
where $(\form{p}^\pm,\vec{q}_\pm)=\cB(\form{p}^{\C\pm},\vec{q}_{\C\pm})$ are the real phase space coordinates defined by the identification map $\cB$ \eqref{eq:t-identification}. 
This shows that \eqref{eq:project} together with \eqref{eq:real_symp_tranf}, defining the subspace $M_S$ \eqref{eq:M_S} of the real phase space, is equivalent to ${\mathcal H}_S(\form{p}^{\C-},\vec{q}_{\C-})=(\form{p}^{\C+},\vec{q}_{\C+})$. 
\end{proof}

It is also useful to restate Lemma \ref{lm:partial-Hessian} in complex coordinates. We introduce for an action $S(\vec{q}_\C,\vec{Q}_\C)$
\begin{equation}\label{eq:S_o-S_red-complex}
S_o(\vec{q}_\C,\vec{Q}_\C,\form{\lambda}^\C)=S(\vec{q}_\C,\vec{Q}_\C)+\sum_{i=1}^N\Re\left(\lambda_i^\C(Q^i_\C-Q^i_{\C*})\right),\quad S_{red}(\vec{q}_\C,\vec{Q}_{\C*})\, ,
\end{equation}
where $\vec{Q}_{\C*}$ is a constant vector.
We remark that $\Re (z_1z_2)=\Re(z_1)\Re(z_2)-\Im (z_1)\Im (z_2)$, thus the second term in $S_o$ recovers the second term in the $S_o$ action \eqref{eq:S_o}. The following lemma follows directly from Lemma \ref{lm:partial-Hessian}.

\begin{lemma}\label{lm:partial-Hessian-complex}
Let $S$ be an action on $\blue{X^\C\subset}\C^{n}\times \C^N$  and $S_o$, $S_{red}$ are defined in \eqref{eq:S_o-S_red-complex}. Then 
\begin{enumerate}
\item $(\vec{q}_\C,\vec{Q}_\C,\form{\lambda}^\C)\in \Stat(S_o)$ if and only if
\begin{equation}
\vec{Q}_\C=\vec{Q}_{\C\ast},\quad \lambda_i^\C=-\frac{\partial \Re S_o}{\partial Q^i_\C},\ i=1,\ldots, N,\quad \vec{q}_\C\in \Stat(S_{red})\,.
\end{equation}
\item The Hessian ${\bf H}(S_o)$ satisfies
\begin{equation}
\det {\bf H}(S_o)=\det {\bf H}(S_{red})
\label{eq:detHS_0}
\end{equation}
at the corresponding stationary points.
\end{enumerate}
\end{lemma}
Compared to \eqref{eq:detS_o}, the absence of a sign factor in \eqref{eq:detHS_0} is due to the doubling of dimensions when treating complex variables as pairs of real variables, giving $(-1)^{2N}=1$.

\section{Classical analysis of the $\Lambda$-SF model}
\label{sec:classical}

{\Blue The goal of this section is to prove Theorem \ref{thm:geometric} (restated formally as Theorem \ref{thm:geometric-formal}).  
This theorem asserts that the tangent spaces of the two crucial submanifolds -- $\cL_{\coh}$ (encoding boundary data) and $\cL_{M_3}$ (imposed by the bulk flatness) -- intersect transversally at the critical point. Establishing geometric transversality will be used in the proof of non-degeneracy of the Hessian. We summarize our strategy as follows.

The classical Chern-Simons phase space $\cP_\Sigma$ is defined as the moduli space of flat connections on the boundary $\Sigma$ of a 3-manifold $M_3$ (see \eqref{def:P-sigma}). We introduce the map $\iota$ which relates bulk connections to boundary connections. Instead of calculating intersections in the phase space directly, we reduce the problem to a linear algebra question in two steps. Firstly, Lemma \ref{lm:Hessian-1} establishes that the intersection of the tangent spaces is trivial if and only if a specific condition on the variations of the bulk flat connection holds. Specifically, if a variation of the bulk flat connection vanishes on the boundary 4-holed spheres ($\cS_a$), it must vanish everywhere. Then Lemma \ref{lm:non-deg-1} translates Lemma \ref{lm:Hessian-1} into the language of holonomies. It expresses the condition for transversality as a system of linear equations involving $\sl(2,\C)$ Lie algebra elements $u_{ab}$ and the group elements $g_{ab}$ defining the connection.

In Section \ref{subsec:reconstruction}, we identify the intersection points with {\it geometric flat connections} (see Definitions \ref{df:admissible} and \ref{def:geometric}) in the case of suitable {\it geometric boundary data} (Definition \ref{df:boundary-data}). We show how non-degenerate 4-simplices in constant curvature spaces (de Sitter or Anti-de Sitter) naturally generate these connections. This provides the specific values for the group elements $g_{ab}$ used in our linear equations.

Finally, in Section \ref{subsec:proof}, we finish the proof via tangent vector (variation) analysis. More precisely, we solve the linear system established in Lemma \ref{lm:non-deg-1} using the specific geometry of the 4-simplex. Lemma \ref{lm:group2bivectors} relates the holonomy around a face to a rotation generated by a specific bivector (the wedge product of edge vectors). Proposition \ref{prop:U-U} and Lemma \ref{lm:Biv-1} and \ref{lm:Biv-2} are technical results concerning the linear independence of these bivectors in Lorentzian signature. They ensure that the only solution to the linear system derived in Lemma \ref{lm:non-deg-1} is the trivial solution.

By combining these steps, we prove that the geometric constraints of a non-degenerate 4-simplex force the intersection of the Lagrangian submanifolds to be trivial, thereby proving Theorem \ref{thm:geometric}.
}

\subsection{\blue{Classical} Chern-Simons theory on $\Gamma_5$ graph}
\label{sec:CS_phase_space}

 The definition of the vertex amplitude in the $\Lambda$-SF model is based on the Chern-Simons theory for a special graph in the three-sphere $S^3$\blue{, which we now review}. The three-sphere is homeomorphic to the boundary of a $4$-simplex, and this identification provides a cellular decomposition of $S^3$. The $4$-simplex has $5$ vertices ($0$-cells) denoted by $a\in\{1,\ldots, 5\}$ connected by $10$ edges ($1$-cells) that can be labeled by distinct pairs of vertices that they join. Together, this forms a $1$-skeleton of a cellular decomposition (triangulation) -- the graph $\Gamma_5$ (see fig.\ref{fig:Gamma5}). In our analysis, we will also need other elements of this cellular decomposition. There are $5$ tetrahedra $T_a$, $a=1,\ldots, 5$  that form $3$-cells. Each tetrahedron will be labeled by the only vertex $a$ which does not belong to it. The two tetrahedra intersect in a triangle ($2$-cell)
\begin{equation}
T_a\cap T_b
\end{equation}
which are labeled by a pair of distinct vertices. Finally, the intersection of two triangles belonging to a common tetrahedron gives one of the edges of the graph $\Gamma_5$\footnote{Let us remark that, although in the definition of the vertex amplitude we are using $\Gamma_5$ graph, the semiclassical reconstruction is in terms of another triangulation dual to the one described here. For the details, see \cite{Han:2021tzw,Han:2025mkc}.}.

In the definition of the Chern-Simons theory, we need to introduce tubular neighborhood of graph. Consider an open tubular neighborhood $b\Gamma_5$ of $\Gamma_5$. We define
\begin{equation}
M_3:=S^3\setminus b\Gamma_5,\quad \Sigma:=\partial M_3,
\end{equation}
where $\partial$ denotes the boundary. We remark that $\Sigma$  is a genus-6 oriented Riemann surface.

\begin{figure}[h!]
\centering
\includegraphics[width=0.4\textwidth]{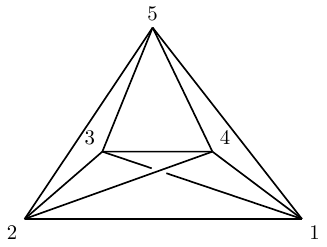}
\caption{
$\Gamma_5$ graph projected on $\R^2$. It forms a triangulation of $S^3$, which is the boundary of a 4-simplex. Numbers $1,\cdots,5$ denote the vertices of the graph. } 
\label{fig:Gamma5}
\end{figure}

The phase space of the Chern-Simons theory is the moduli space $\cM_\Flat(\Sigma,\SL(2,\bC))$ of flat $\SL(2,\bC)$ connections on $\Sigma$, defined as
\be\label{def:P-sigma}
\cP_\Sigma :=\cM_\Flat(\Sigma,\SL(2,\bC))=\Hom(\pi_1(\Sigma),\SL(2,\bC))/\SL(2,\bC)\,,
\ee
where the quotient is by the conjugate action. Except at the thin singular loci, this is a symplectic space of 30 complex dimensions, equipped with the Atiyah-Bott-Goldman symplectic form \blue{\cite{atiyah1983yang, goldman1984symplectic} (see also for explicit formula in FG-FN coordinates \cite{Dimofte:2011gm,Han:2021tzw})}. \blue{It depends on a choice of so-called Chern-Simons level $s\in \C^*$
\begin{equation}
\Omega=\frac{1}{4\pi}\left(s\Omega^\C+\overline{s}\overline{\Omega^\C}\right)
\end{equation}
where $\Omega^\C$ is a complex symplectic form. The details of this form will not be important in our argument, so we postpone some details to the later part of the paper.}

We can pull back the flat connection from $M_3$ to $\Sigma$ (which is a boundary of $M_3$), and this operation is covariant with respect to the gauge transformations, thus we obtain a map 
\begin{equation}
\iota\colon \cM_{\Flat}(M_3,\SLCC)\rightarrow \cM_{\Flat}(\Sigma,\SLCC)\,.
\label{eq:Omega_CS}
\end{equation}
The image of this map is Lagrangian, meaning that
$\iota^\ast\Omega=0$ (the pull-back of sympletic form vanishes). In a general situation, $\iota$ might not be an embedding. In our case, however, the map satisfies this assumption and we denote ${\mathcal L}_{M_3}$ the corresponding image. It is a Lagrangian submanifold on the smooth locus of $\cM_{\Flat}(\Sigma,\SLCC)$ consisting of flat connections that can be obtained by \blue{ a pull-back of} a flat connection on $M_3$ to $\Sigma$.

Let us introduce a small ball $V_a$ around vertex $a$ of $\Gamma_5$ for $a=1,\ldots, 5$. The intersection 
\begin{equation}
\cS_a :=V_a\cap \Sigma
\end{equation}
is a 4-holed sphere. Surface $\Sigma=\partial M_3$ is composed of five 4-holed spheres $\cS_a (a=1,\cdots,5)$ and 10 annuli $(ab)$'s with $a,b=1,\ldots,5$, $a<b$ each connecting a pair of holes from $\cS_a$ and $\cS_b$.

The space of flat connection on $\cS_a$ will be denoted by
\begin{equation}
\cM^0_\Flat(\cS_a,\SL(2,\bC))=\Hom(\pi_1(\cS_a),\SL(2,\bC))/\SL(2,\bC)\,.
\end{equation}
We can restrict flat connections from $\Sigma$ to $\cS_a$, and this operation is also covariant with respect to gauge transformations, thus we obtain
\begin{equation}
\pi_{\cS_a}\colon \cM_\Flat(\Sigma,\SL(2,\bC))\rightarrow \cM^0_\Flat(\cS_a,\SL(2,\bC))\,.
\end{equation}
Moreover, $\pi_{\cS_a}\iota$ can be described as a restriction of flat connection on $M_3$ to $\cS_a$. The space $\cM^0_\Flat(\cS_a,\SL(2,\bC))$ does not possess a natural symplectic structure.

{\Blue
Consider a direct product space $\cM^0$ together with a map from $\cP_\Sigma$:
\begin{equation}
 \cM^0:=\bigtimes_{a=1}^5\cM^0_{\Flat}(\cS_a,\SLCC) , \quad \pi=\bigtimes_{a=1}^5\pi_{\cS_a}\colon \cP_\Sigma\rightarrow \cM^0\,.
 \label{eq:def_M0_pi}
\end{equation}
For $m=(m_1,\ldots,m_5)\in \cM^0$, we define a subset of $\cP_\Sigma$
\begin{equation}\label{def:coh-Lamgr}
{\mathcal L}_{\coh}(m):=\{x\in \cP_\Sigma\colon \pi_{\cS_a}(x)=m_a,\ \forall\, a=1,\ldots, 5\}\,.
\end{equation}
A few remarks on this subset follows. 
\begin{enumerate}
    \item It is non-empty if and only if the conjugacy class of holonomies around the corresponding holes (the ones that wind around the same edge of $\Gamma_5$) of the $4$-hole spheres match.
    \item It is a submanifold for generic $m$. For example, this holds if every $m_a$ belongs to an open subset described by the Fock-Goncharov coordinates (see \cite{Dimofte:2011ju})\footnote{\blue{
We will see in Section \ref{sec:critical_spinfoam} that $\cL_{\coh}(m)$  is a submanifold of $\cP_\Sigma$ using the FG-FN coordinate description. 
Essentially, $\cL_{\coh}(m)$ corresponds to the subspace of flat connections where the geometry (captured by variables $(L_{ab},X_a,Y_a)$, described in Section \ref{subsubsec:FGFN}) of each individual 4-holed sphere $\cS_a$ is fixed, but the ``twisting" relative to each other (captured by the gluing data $T_{ab}$) is allowed to vary freely.}}. In particular, the boundary data considered in this paper (see Definition \ref{df:boundary-data}) satisfy this requirement.
\end{enumerate}
In what follow, we assume that $\cL_{\rm coh}(m)$ is a submanifold of $\cP_{\Sigma}$ (thus, in particular, it is disjoint to singular loci of $\cP_{\Sigma}$).
It depends on a choice of $m$, but we will often skip this label and use ${\mathcal L}_{\coh}$.}
Our approach to proving Theorem \ref{thm:geometric} will be through the following lemma:

\begin{lemma}\label{lm:Hessian-1}
\blue{Let $x\in {\mathcal L}_{\coh}\cap {\mathcal L}_{M_3}$ and let $y\in \cM_{\Flat}(M_3,\SLCC)$ be such that $x=\iota(y)$.} Suppose that the following is true
\begin{equation}
\{v\in T_{y}\cM_{\Flat}(M_3,\SLCC)\colon D_{y}(\pi_{\cS_a}\iota)(v)=0, \, \forall \, a=1,\ldots, 5\}=\{0\}\,.
\label{eq:v=0}
\end{equation}
Then
\begin{equation}
\blue{T_x{\mathcal L}_{\coh}\cap T_x{\mathcal L}_{M_3}=\{0\}\,.}
\end{equation}
\end{lemma}

\begin{proof}
{\Blue The map $\iota$ is a diffeomorphism thus $w\in T_x{\mathcal L}_{M_3}$ if there exists a $v\in T_{y}\cM_{\Flat}(M_3,\SLCC)$ such that $w=D_{y}\iota(v)$. On the other hand, if $w\in T_x{\mathcal L}_{\coh}$ then, by definition \eqref{def:coh-Lamgr}\footnote{In fact, as we will see in Section \ref{sec:critical_spinfoam}, it is if and only if. See particularly Lemma \ref{lm:isom-FG} and Remark \ref{rmk:iso}.}
\begin{equation}
D_x\pi_{\cS_a}(w)=0,\quad \forall\,a=1,\ldots, 5.
\end{equation}
Suppose that $w\in T_x{\mathcal L}_{\coh}\cap T_x{\mathcal L}_{M_3}$ then
\begin{equation}
D_{y}(\pi_{\cS_a}\iota)(v)=D_{x}\pi_{\cS_a}D_{y}\iota(v)=D_{x}\pi_{\cS_a}(w)=0
\end{equation}
However, due to the assumption \eqref{eq:v=0}, this means that $v=0$ and, as a consequence, $w=D_{y}\iota(v)=0$.}
\end{proof}
\blue{This lemma confirms that, if fixing the connection on all the 4-holed spheres $\cS_a$ forces the variation inside the bulk $v$ to vanish, then the geometry is rigid.
In order to apply this result, we need a convenient description of the tangent space to $\cM_{\Flat}(M_3,\SLCC)$ and of the maps $\pi_{\cS_a}\iota$.}

\subsection{Holonomy description}
\label{subsec:flat}

 The above analysis shows that the question of non-degeneracy of the Hessian can be answered by analyzing properties of vectors in the tangent space to the moduli space of flat connections on $M_3$ and its projections on the space of flat connections on the $4$-hole spheres. In order to utilize this observation, we need a convenient description of these spaces. 

We first describe flat connections on a $d$-dimensional ($d\geq 2$) smooth manifold $N$, possibly with a boundary. Introduce a cellular decomposition of $N$, where a cellular complex consists of contractible closed cells. We consider only the case when, for any $0<n\leq d$, the intersection of  two $n$-dimensional cells is a disjoint sum of $(n-1)$-dimensional cells.

On every $d$-cell, we can choose a gauge in which the connection is trivial. The gauge choice is not unique,  but any two such trivializations are related by a constant gauge transformation.  In particular, for every $(d-1)$-cell in the intersection of two $d$-cells, we have two different gauges of these two cells that are related by a constant group element. It is convenient to introduce the orientation of a $(d-1)$-cell. The orientation allows us to distinguish between two $d$-cells separated by a $(d-1)$-cell into the initial and final $d$-cell. The group element associated to the oriented $(d-1)$-cell is given by the change of gauge from the trivialization on the initial cell to the trivialization on the final cell. For the same $(d-1)$-cell but with an opposite orientation, the group element is the inverse of the other one. Every $(d-2)$-cell imposes some consistency condition on the group elements associated to the oriented $(d-1)$-cells\blue{:} changing gauges in a cyclic order around a $(d-2)$-cell should give, after closing the loop, identity. This means that the cyclic product of group elements of $(d-1)$-cells sharing the same $(d-2)$-cell should be equal to identity. This is called the closure condition. These are the only conditions on the group elements associated to oriented $(d-1)$-cells to define a flat connection on $N$. It is not surprising as the curvature is associated to the $(d-2)$-cells.

The discussion above allows us to describe the space of flat connections on a $d$-dimensional manifold $N$ with boundary as follows\footnote{\blue{Description of flat connections in terms of cocycles is standard, see for example \cite{guichard2018introduction}.}}. 
Let ${\mathcal C}_d(N)$ be the set of $d$-cells in the chosen cellular decomposition of $N$, ${\mathcal C}_{d-1}^o(N)$ be the set of oriented $(d-1)$-cells and ${\mathcal C}_{d-2}(N)$ set of $(d-2)$-cells. 
For every $e\in {\mathcal C}_{d-1}^o(N)$, we have initial $d$-cell $i(e)$  and final $d$-cell $f(e)$  of $e$. Moreover, $e^{-1}$ is the $(d-1)$-cell with the reverse orientation. Let
\begin{equation}
\Hol_{\Flat}(N):=\left\{(g_e)\in \SLCC^{{\mathcal C}_{d-1}^o(N)}\colon \forall\, e\in {\mathcal C}_{d-1}^o(N),\ g_{e^{-1}}=g_e^{-1}\,; \,\forall\, 
f\in{\mathcal C}_{d-2}(N),\ \prod_{e\supset f}^{\rightarrow} g_e=1\right\}\,.
\label{eq:Hol_flat-def}
\end{equation}
where we denoted $e\supset f$ if $(d-2)$-cell $f$ belong to $(d-1)$ cell $e$. The gauge action of $\SLCC^{{\mathcal C}_d(N)}$ on $\Hol_{\Flat}$ is by
\begin{equation}
(h_v)\cdot (g_e)=(g_e'=h_{f(e)}g_e h_{i(e)}^{-1})\,,\quad (h_v)\in \SLCC^{{\mathcal C}_d(N)}\,,\quad (g_e)\in \Hol_{\Flat}(N)\,.
\end{equation}
The moduli space of flat connections on $N$ is described by 
\begin{equation}
\cM_{\Flat}(N,\SLCC)=\Hol_{\Flat}(N)/\SLCC^{{\mathcal C}_d(N)}\,.
\end{equation}

We can now describe the tangent vectors at the smooth loci of this space in terms of infinitesimal variations of $g_e$, $e\in {\mathcal C}^o_{d-1}(N)$.
Every vector $t$ at $(g_e)\in \SLCC^{{\mathcal C}_{d-1}^o(N)}$ can be described by matrices $\delta_t g_e$ satisfying $g_e^{-1}\delta_t g_e\in \slcc$. For it to be tangent to $\Hol_{\Flat}(N)$, it must preserve the constraints in \eqref{eq:Hol_flat-def}. That is,
\begin{equation}
\delta_t(g_eg_{e^{-1}})=0,
\quad \delta_t\left(\prod_{e\supset f}^{\rightarrow} g_e\right)=0,
\end{equation}
where we understand the conditions in terms of Leibniz rules. A vector is trivial if it is tangent to a gauge transformation. This means that
\begin{equation}
t=0\Longleftrightarrow \exists \,(u_v)\in \slcc^{{\mathcal C}_d(N)}\colon \delta_tg_e=u_{f(e)}g_e-g_eu_{i(e)}\,.
\end{equation}

Let us focus on the case when $N=M_3$ ($d=3$).  Our cellular decomposition of $S^3$ provides a cellular decomposition of $M_3$. Every cell of this decomposition is obtained by the intersection of a cell from the cellular decomposition of $S^3$ with $M_3$. In particular, the $3$-cells $\tilde{T}_a$, $a=1,\ldots, 5$, are defined as\footnote{\blue{As we will see in Section \ref{subsubsec:FGFN}, $\tilde{T}_a$ is in fact the ideal octahedron $\Oct(a)$ which play pivotal role in quantization of the theory \cite{Han:2021tzw}.}}
\begin{equation}
\tilde{T}_a:=T_a\cap M_3\,,
\label{eq:T_tilde}
\end{equation}
and it is illustrated in fig.\ref{fig:oct_Ta}. 
\begin{figure}[h!]
\centering
\includegraphics[width=0.4\textwidth]{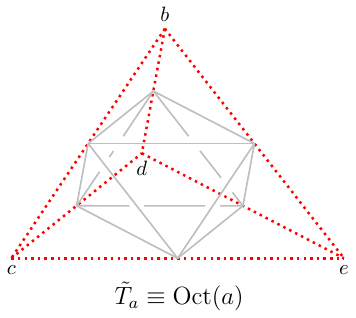}
\caption{\blue{3-cell $\tilde{T}_a$. Here, $a,b,c,d,e$ are distinct. It is the tetrahedron $T_a$ with tubular open neighhorbood of edges ({\it in dotted red}) removed. $\tilde{T}_a$ also coincides with the ideal octahedron $\Oct(a)$ defined in \cite{Han:2021tzw} and reviewed in Section \ref{subsubsec:FGFN}. Here, the cusp boundaries of $\Oct(a)$ are shrunk to vertices of the octahedron ({\it in gray}). See \cite{Han:2021tzw} for more details.} }
\label{fig:oct_Ta}
\end{figure} 
They intersect in oriented $2$-cells
\begin{equation}
\tilde{F}_{ab}:=\tilde{T}_a\cap \tilde{T}_b=T_a\cap T_b\cap M_3\,,
\end{equation}
where the orientation is such that $\tilde{T}_b$ is the initial cell and $\tilde{T}_a$ is the final cell. The set of $1$-cells is empty, because $\Gamma_5\cap M_3=\emptyset$. In summary,
\begin{equation}
{\mathcal C}_{d}(M_3)=L_1,\quad {\mathcal C}_{d-1}^o(M_3)=L_2,\quad {\mathcal C}_{d-2}(M_3)=\emptyset,
\end{equation}
where we introduced the sets $L_1$ (of five $3$-cells) and $L_2$ (of the twenty oriented 2-cells),
\begin{equation}
L_2=\{(a,b)\in\{1,\ldots, 5\}^2, a\not=b\}\text{ and }L_1=\{1,\ldots, 5\}\,.
\end{equation}
 We can now use our description of the flat connections to describe $\Hol_{\Flat}(M_3)$. Then the moduli space of flat connections on $M_3$ can be described as follows.
\begin{equation}
\begin{split}
\Hol_{\Flat}(M_3)=\{(g_{ab})\in\SLCC^{L_2}\colon g_{ab}=g_{ba}^{-1}\}\,,\\
\cM_{\Flat}(M_3,\SLCC)=\Hol_{\Flat}(M_3)/\SLCC^{L_1}\,.
\end{split}
\label{eq:def_holFlatM3}
\end{equation}
Notice that, due to ${\mathcal C}_{d-2}(M_3)$ being an empty set, there is {\it no} closure condition. 
For the generic point of $\Hol_{\Flat}(M_3)$, the stabilizer of the action of $\SLCC^{L_1}$ is discrete, thus the complex dimension of the smooth loci of the set $\cM_{\Flat}(M_3,\SLCC)$ is equal $20/2*3-5*3=15$ as expected.

Similar construction can be made for every $\cS_a$  ($d=2$). Let us fix $a$. 
From the cellular decomposition of $M_3$, we can obtain a cellular decomposition of $\cS_a$. It is done by intersecting cells from $M_3$ with $\cS_a$.
We introduce $2$-cells $F_b$ for $b\not=a$ obtained by intersecting $3$-cells of $M_3$ with $\cS_a$:
\begin{equation}
F_b:=\tilde{T}_b\cap \cS_a\,.
\end{equation}
It touches three holes out of four of $\cS_a$. 
An illustration of $F_b$ and $E_{bc}$ is given in fig.\ref{fig:Fb_Ebc}.
\begin{figure}[h!]
\centering
\begin{minipage}{0.45\textwidth}
\centering
\includegraphics[width=0.6\textwidth]{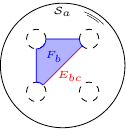}
\end{minipage}
\begin{minipage}{0.45\textwidth}
\centering
\includegraphics[width=0.6\textwidth]{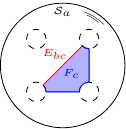}
\end{minipage}
	\caption{2-cells $F_b$ and $F_c$ on a fixed $\cS_a$ ({\it shaded in blue}). Each of them touches three out of four holes of $\cS_a$. Their intersection gives 1-cell $E_{bc}$ ({\it in red}).} 
	\label{fig:Fb_Ebc}
\end{figure}
The set of $1$-cells is obtained through intersecting pairs of $2$-cells. We introduce oriented $1$-cells
\begin{equation}
E_{bc}:=F_b\cap F_c=\tilde{F}_{bc}\cap \cS_a,
\end{equation}
with $b\not=c$ and $b,c\not=a$. \blue{An illustration of $F_b$ and $E_{bc}$ is given in fig.\ref{fig:Fb_Ebc}.} The orientation is such that $F_c$ is the initial 2-cell and $F_b$ is the final 2-cell.  For the same reason as in the case of $M_3$, the intersection of $1$-cells is always empty. We introduce
\begin{equation}
L_2^a=\{(b,c)\in\{1,\ldots, 5\}^2, b\not=c,\ b\not=a,\ c\not=a\}\text{ and }L_1^a=\{1,\ldots, 5\}\setminus \{a\}.
\end{equation}
The moduli space of flat connections on $\cS_a$ is given by
\begin{equation}
 \cM^0_{\Flat}(\cS_a,\SLCC)=\Hol_{\Flat}(\cS_a)/\SLCC^{L_1^a},\quad 
\Hol_{\Flat}(\cS_a)=\{(g_{bc})\in\SLCC^{L_2^a}\colon g_{bc}=g_{cb}^{-1}\}
\,.
\end{equation}
Again, the absence of the closure condition is due to ${\mathcal C}_{d-2}(\cS_a)=\emptyset$. Similarly to the case of $\cM_{\Flat}(M_3,\SLCC)$, one can show that the complex dimension of the smooth loci of this space is equal to 6, as expected.

As the cellular decomposition of $\cS_a$ is obtained from the cellular decomposition of $M_3$, our representation of flat connections allows for a simple description of the restriction of flat connections on $M_3$ to $\cS_a$. It has a very simple representation:
\begin{equation}
\pi_{\cS_a}\iota((g_{bc}))=(g_{bc})_{b,c\not=a}.
\end{equation}

\begin{lemma}\label{lm:non-deg-1}
Let $(g_{ab})\in \cM_{\rm flat}(M_3,\SLCC)$  be such that its image $x$ under $\iota$ belongs to ${\mathcal L}_{\coh}$. Consider two conditions for a tangent vector $t$ at this point of $\cM_{\Flat}(M_3,\SLCC)$:
\begin{enumerate}
\item\label{lm:1-cond} There exist $u_{ab}\in \slcc$, $a\not=b$, such that
\begin{equation}
\delta_t g_{ab}=u_{ca}g_{ab}-g_{ab}u_{cb}
\end{equation}
for every $a$, $b$, $c$ that are pairwise different;
\item\label{lm:2-cond} $t=0$ i.e. there exist $u_{c}\in \slcc$, $c=1,\ldots ,5$, such that
\begin{equation}
\delta_t g_{ab}=u_{a}g_{ab}-g_{ab}u_{b}
\end{equation}
for every $a\not=b$.
\end{enumerate}
If $\eqref{lm:1-cond}\Longrightarrow \eqref{lm:2-cond}$, then {\Blue $T_{x}{\mathcal L}_{\coh}\cap T_x{\mathcal L}_{M_3}=\{0\}$}.
\end{lemma}

\begin{proof}
It is the restated condition from Lemma \ref{lm:Hessian-1} using the description of tangent vectors to spaces of flat connections. 
The variation $\delta_t$ realizes the tangent vector $v'$ in Lemma \ref{lm:Hessian-1}.
The first point describes the vanishing of the tangent vector projected by $\pi_{\cS_a}\iota$ for every $a$. The second point describes the vanishing of the tangent vectors in the space of flat connections on $M_3$.
\end{proof}
 
 This lemma reduces {\Blue the proof of Theorem \ref{thm:geometric}}  to a purely combinatorial problem\blue{, centering at studying the $\SL(2,\bC)$ group elements defining the moduli space $\cM_{\rm flat}(M_3,\SLCC)$ of flat connections and the tangent vectors on it.}
 We will now analyze this problem in the case when group elements $g_{ab}$ are obtained from a stationary point that corresponds to a non-degenerate $4$-simplex.

\subsection{Reconstruction of 4-simplex geometry}
\label{subsec:reconstruction}

We now move to  {\Blue the intersection points of ${\mathcal L}_{\coh}\cap {\mathcal L}_{M_3}$ produced by non-degenerate 4-simplex geometry.} Our main tool will be Lemma \ref{lm:non-deg-1}. In order to apply it to the $\Lambda$-SF model, we need to translate the original description of the stationary points from \cite{Haggard:2014xoa,Haggard:2015ima,Haggard:2015nat} into our description of the flat connections on $M_3$ as described in Section \ref{subsec:flat}.

Under mild non-degeneracy conditions, the {\Blue intersection points} can be described in the following way. Consider a homogeneously curved non-degenerate Lorentzian 4-simplex with spacelike tetrahedra and hence also spacelike triangles, whose curvature can be positive or negative. Non-degeneracy means that all tetrahedra are non-degenerate and any four tetrahedron normals at their common vertex are linearly independent. After choosing some spin frames at the vertices, $g_{ab}$ in $\Hol_{\Flat}(M_3)$  is given by spin parallel transport from vertex $b$ to vertex $a$ along the edge of the $4$-simplex connecting these vertices. This describes an {\Blue intersection} point in $\cM_{\Flat}(M_3,\SLCC)$.


We first consider the case of a positive cosmological constant. 
De Sitter space is a hypersurface of $\R^{1,4}$:
\begin{equation}
dS=\{X\in \R^{1,4}\colon \eta_{IJ}X^IX^J=-1\}\,,
\end{equation}
where $\eta=(\rd X^0)^2-\sum_{I=1}^4 (\rd X^I)^2$. Metric $\eta$ restricted to $dS$ gives the de Sitter metric. The tangent space at $Y\in dS$ can be identified with
\begin{equation}\label{eq:tangent}
T_YdS=\{v\in \R^{1,4}\colon \eta_{IJ} Y^Iv^J=0\}\,.
\end{equation}
Suppose two distinct points $X,Y\in dS$ can be connected by a spacelike geodesic of length smaller than $\pi$. This geodesic can be determined as follows:
There exists a unique two-dimensional plane $H\subset \R^{1,4}$ containing the origin such that $X,Y\in H$. This plane is spacelike and $H\cap dS$ is the unique geodesic circle to which $X,Y$ belong. The shorter segment of this circle is the geodesic we are searching for. It is easy to describe the (non-normalized) initial velocity $\gamma$ of this geodesic at point $X$ using \eqref{eq:tangent}:
\begin{equation}
\gamma=Y+(Y\cdot X)X,\quad Y\cdot X=\eta_{IJ}X^IY^J\,.
\end{equation}
The description of geodesics can be extended to totally geodesic surfaces in de Sitter.  

Let $X_\alpha\in dS$ for $\alpha=1,\ldots k+1$ be $k+1$ points on de Sitter. Suppose that $X_\alpha$ as vectors in $\R^{1,4}$ are independent and every two of them can be connected by a geodesic. Let 
\begin{equation}
H=\operatorname{span}\{X_\alpha,\ \alpha=1,\ldots k+1\}
\end{equation}
be a subspace in $\R^{1,4}$. Then the connected component $N$ of $H\cap dS$ containing $X_\alpha$, $\alpha=1,\ldots, k+1$, is the unique $k$-dimensional, totally geodesic and connected submanifold of $dS$ containing all all the points $X_\alpha$, $\alpha=1,\ldots, k+1$. The tangent space of $N$ at $X_\alpha$ is spanned by the initial velocities of the geodesics connecting this point with $X_\beta$, $\beta\not=\alpha$, \ie
\begin{equation}
\gamma_{\beta\alpha}=X_\beta+(X_\beta\cdot X_\alpha)X_\alpha\,.
\end{equation}
Let us notice that $\gamma_{\beta\alpha}$ for $\beta\in \{1,\ldots k+1\}\setminus\{\alpha\}$ are linearly independent due to the linear independence of $X_\alpha$, $\alpha\in \{1,\ldots k+1\}$.

 The same construction can be done for anti-de Sitter space, but on the ambient space $\R^{2,3}$ with signature $(++---)$. 
\begin{equation}
AdS=\{X\in \R^{2,3}\colon \eta'_{IJ}X^IX^J=1\}\,,
\end{equation}
where $\eta'=(\rd X^0)^2+(\rd X_1)^2-\sum_{I=2}^4 (\rd X^I)^2$. Metric $\eta'$ restricted to $AdS$ gives the Anti-de Sitter metric. The difference is that if two points can be connected by a spacelike geodesic then it is unique.

We can now describe non-degenerate $4$-simplices with spacelike tetrahedra:

\begin{definition}\label{df:admissible}
The set of points $X_a\in dS$ (or $X_a\in AdS$), $a\in\{1,\ldots, 5\}$ is {\bf admissible} if
\begin{enumerate}
\item $X_a$, $a\in\{1,\ldots, 5\}$, are linearly independent,
\item for every $a$, $b$ distinct, $X_a$ and $X_b$ can be connected by a geodesic
\item for every $a$,
\begin{equation}
H_a=\operatorname{span}\{X_b,\ b\in\{1,\ldots, 5\}\setminus\{a\}\}
\end{equation}
is a spacelike subspace of $\R^{1,4}$  (dS) or Lorentzian signature subspace (AdS).
\end{enumerate}
\end{definition}

Let $X_a\in dS$  (or $X_a\in AdS$), $a=1,\ldots, 5$ be admissible. The unique short geodesic connecting distinct points provides the edges of $4$-simplex. Let us choose an orthonormal oriented {\Blue and time-oriented} frame at every vertex. Parallel transport from vertex $b$ to vertex $a$ determines a group element $G_{ab}\in \SO_+(1,3)$. 
The change of frames is by the action of $G_a\in \SO_+(1,3)$ at vertex $a$, transforming as
\begin{equation}\label{eq:gauge-SO}
G_{ab}\rightarrow G_a G_{ab}G_b^{-1}\,.
\end{equation}
Let $\Pi$ be a homomorphism of $\SLCC$ onto the Lorentz group $SO_+(1,3)$. 
\blue{\begin{definition}
\label{def:geometric}
An element $(g_{ab})\in \cM_{\Flat}(M_3,\SLCC)$ is called a {\bf geometric flat connection} (in $\cM_{\Flat}(M_3,\SLCC)$) (in short $(g_{ab})$ is geometric) if there exists an admissible set of points $X_a\in dS$ (or $X_a\in AdS$), $a\in\{1,\ldots, 5\}$ and a choice of frame at the corresponding $4$-simplex such that
\begin{equation}
\Pi(g_{ab})=G_{ab}.
\end{equation}
\blue{Every geometric flat connection is an image under $\iota$ of a unique point in $\cM_{\Flat}(\Sigma,\SLCC)$. When no confusion arises, we also call this element of $\cM_{\Flat}(\Sigma,\SLCC)$ a geometric flat connection.}
\end{definition}}

{\Blue
First, observe that if a set of group elements $(g_{ab})$ satisfies the definition of being geometric, then any gauge equivalent set also satisfies it (due to the transformation property in \eqref{eq:gauge-SO}). Therefore, Definition \ref{def:geometric} indeed applies  to elements of the moduli space $\cM_{\Flat}(M_3,\SLCC)$, not just to specific representatives in $\Hol_{\Flat}(M_3)$. 
It is shown in \cite{Haggard:2015ima} that a geometric flat connection uniquely determines a set of admissible points $X_a$ (representing a curved 4-simplex). This geometry is unique up to orientation-preserving isometries of $dS$ (or $AdS$) space, with an additional ambiguity regarding parity reversal ($X_a\rightarrow -X_a$ for all $a$ simultaneously). We refer to the equivalence class of such points as the {\bf reconstructed $4$-simplex} for geometric flat connection $(g_{ab})$.

Conversely, consider an admissible set of points $X_a$, $a=1,\ldots, 5$ determining a fixed 4-simplex geometry (up to isometry). While the Lorentz group elements $G_{ab}$ are fixed by this geometry, their lift to $\SLCC$ is not unique. One can verify that there are exactly $2^{6}$ distinct geometric flat connections corresponding to the same 4-simplex geometry \cite{Haggard:2015ima}\footnote{
\blue{The number of distinct flat connections can be calculated as follows. 10 edges in $\Gamma_5$ give totally $2^{10}$ raw combinations of lifts. We can perform a gauge transformation at the vertices. A sign change at a vertex $a$ flips the signs of all edges connected to $a$, giving $2^5$ gauge choices. However, a global sign change at all vertices simultaneously does not change the relative connection on the edges. Therefore, the effective gauge freedom is $2^{5-1}=2^4$. This gives the number of distinct flat connections: $2^{10}/2^4=2^6$. }
}.}

Finally, we emphasize the dual nature of this construction. In our description of flat connections on $M_3$, the group elements are associated with 2-cells. Consequently, the reconstructed geometric 4-simplex is dual to the graph $\Gamma_5$: the edges of the geometric 4-simplex correspond to the 2-cells of the cellular decomposition, and the vertices of the geometric 4-simplex correspond to the 3-cells of the complex.

{\Blue
We remind that $\cM^0=\bigtimes_{a=1}^5 \cM_{\Flat}(\cS_a,\SLCC)$ and $\pi=\bigtimes_{a=1}^5 \pi_{\cS_a}$ (see Section \ref{sec:CS_phase_space}). 

\begin{definition}\label{df:boundary-data}
An element $m\in \cM^0$ is called {\bf geometric boundary data} if it is an image under $\pi\iota$ of a geometric flat connection in $\cM_{\Flat}(M_3,\SLCC)$.
\end{definition}

A few remarks are in order:
\begin{enumerate}
\item Consider a single component of the boundary data, $m_a\in \cM^0_{\Flat}(\cS_a,\SLCC)$.Suppose that there exists a gauge in which all group elements belong to $\SU(2)$, but there is no gauge in which all of them belong to the abelian subgroup ${\rm U}(1)$ (which would imply a degenerate geometry). It was shown in \cite{Haggard:2015ima} (see also \cite{Haggard:2015sl}) that, by analogous construction to the one described in this section, $m_a$ defines the geometry of flatly embedded tetrahedron in a constantly curved space (sphere or hyperboloid), or, in some cases, $m_a$ corresponds to a ``two-sheeted hyperbolic tetrahedron" (see \cite{Haggard:2015ima} for the definition). In the former case, we call $m_a$ geometric.
\item Suppose $m=(m_1,\ldots m_5)\in \cM^0$ is such that every $m_a$ is geometric (related to a tetrahedron as described above). The set of $y\in \cM_{\Flat}(M_3,\SLCC)$ such that $\pi\iota(y)=m$ was analyzed in \cite{Haggard:2015sl}. If this pre-image set is non-empty, it consists of either one or two-points. In the latter case, both points have the interpretation of either degenerate so-called vector geometries or they are geometric flat connections of two constantly curved $4$-simplices (or ``two-sheeted 4-simplices") related by a parity reflection (orientation reversal).
\end{enumerate}

In this paper, we restrict our analysis to the ``good" sector:  the geometric boundary data (Definition \ref{df:boundary-data}) correspond to standard tetrahedra and a reconstructed non-degenerate, curved 4-simplex. We will briefly comment on the remaining cases in the Conclusion.

We can now state formally Theorem \ref{thm:geometric}. 

\begin{theorem}[Theorem \ref{thm:geometric}]\label{thm:geometric-formal}
Let $m\in \cM^0$ be a geometric boundary data and let $y$ be such that $\pi\iota(y)=m$. Then
\begin{equation}
T_{\iota(y)}{\mathcal L}_{\coh}(m)\cap T_{\iota(y)}{\mathcal L}_{M_3}=\{0\}.
\end{equation}
where $\iota(y)\in {\mathcal L}_{\coh}(m)\cap {\mathcal L}_{M_3}$.
\end{theorem}
}

\subsection{Proof of Theorem \ref{thm:geometric}}
\label{subsec:proof}

We are now ready to prove {\Blue Theorem \ref{thm:geometric}}.
Using our description of flat connection, we reduce a question about the intersection to a question about a \blue{couple} of $\SLCC$ group elements and $\slcc$ Lie algebra elements (assumptions of Lemma  \ref{lm:non-deg-1}). We will now use the properties of the flat connections corresponding to non-degenerate $4$-simplices that were derived in the previous section to show that the assumptions of Lemma 
\ref{lm:non-deg-1} hold for such 4-simplices. 

Firstly, we need to determine some properties of holonomies around the faces of a $4$-simplex in de Sitter and anti-de Sitter spaces.
Recall that we have chosen a frame at vertex $a$. It gives an identification of the tangent space to the de Sitter or anti-de Sitter space at the vertex with $\R^{1,3}$. The parallel transport around the face $abc$ is given in this frame by
\begin{equation}
G_{cba}:=G_{ac}G_{cb}G_{ba}\in \SO_+(1,3)\,.
\end{equation}
Geodesic connecting vertex $a$ with $b$ and vertex $a$ with $c$ have tangent vectors at $a$ given in the frame by
\begin{equation}
\gamma_{ba},\gamma_{ca}\in \R^{1,3}\,.
\end{equation}
We can state some basic properties of this holonomy for a non-degenerate $4$-simplex: 

\begin{lemma}
\label{lm:group2bivectors}
For every distinct $a$, $b$ and $c$ in a non-degenerate $4$-simplex in dS or AdS,
\begin{equation}
G_{cba}=e^{\tau B},\quad B=\gamma_{ba}\wedge\gamma_{ca}\,,
\end{equation}
and moreover, $G_{cba}\not=1$.
\end{lemma}

\begin{proof}
Let $H$ be a three-dimensional hyperplane containing the origin and $X_a,X_b,X_c$. The two-dimensional submanifold $N=H\cap dS$ (or $N=H\cap AdS$ in case of negative cosmological constant) is totally geodesic. This means that the parallel transport preserves the normal vectors to this hypersurface.

Vectors $\gamma_{ba}$, $\gamma_{ca}$ span the tangent space to $N$ at vertex $a$ (they are independent as the $4$-simplex is non-degenerate). This means that $G_{cba}$ is a rotation in the plane spanned by $\gamma_{ba}$ and $\gamma_{ca}$, and the vectors orthogonal to this plane are preserved by $G_{cba}$. Thus
\begin{equation}
G_{cba}=e^{\tau B},\quad B=\gamma_{ba}\wedge\gamma_{ca}.
\end{equation}
In order to determine whether $G_{cba}=1$, we can restrict the problem to $N$ that is either a sphere (for $dS$) or a hyperbolic 2-plane (for $AdS$).

In two dimensions, $G_{cba}$ is given by rotation by an angle equal to $\pm$ area of the triangle. In spherical geometry, proper triangles have areas less than $2\pi$, while in hyperbolic geometry, the proper triangles have areas less than $\pi$. So $G_{cba}\not=1$.
\end{proof}

We will now prove that the assumptions of Lemma \ref{lm:non-deg-1} are satisfied for $(g_{ab})$ coming from the non-degenerate $4$-simplex. 

\begin{prop}\label{prop:U-U}
Let $(G_{ab})$ be parallel transports for $4$-simplex obtained by an admissible set of vertices either in de Sitter or anti-de Sitter. Suppose that the bivectors
\begin{equation}
U_{ab}\,\in {\bigwedge} ^2 \R^{1,3}\,,\quad (a,b)\in L_2
\end{equation}
satisfy, for distinct $a$, $b$, $c$, $d$ and $d'$, that
\begin{equation}
\Ad_{G_{cba}}U_{da}-U_{da}=\Ad_{G_{cba}}U_{d'a}-U_{d'a}\,.
\end{equation}
Then there exist bivectors $U_a$, $a\in L_1$ such that $U_{da}=U_a$ for every $d\not=a$.
\end{prop}

We will base the proof on some properties of bivectors. 

\begin{lemma}\label{lm:Biv-1}
Suppose the bivector $U \in {\bigwedge} ^2 \R^{1,3}$ satisfies
\begin{equation}
\Ad_{e^{B}}U=U
\end{equation}
for a simple\footnote{A bivector $B$ is called simple if there exist two vectors $u,v$ such that $B=u\w v$. } spacelike bivector $B$ such that $e^{B}\not=1$. Then
\begin{equation}
U=\alpha B+\beta \ast B
\end{equation}
with $\alpha,\beta\in\R$. 
\end{lemma}

\begin{proof}
As $B$ is simple and spacelike, there exists a unit timelike vector $n$ such that
\begin{equation}
n\llcorner B=0\,,
\end{equation}
and so $e^Bn=n$. We introduce a subspace $V=\{v\in \R^{1,3}\colon v\cdot n=0\}$. It is an Euclidean subspace $\R^3$. Moreover, we can regard $e^B$ as an element in $\SO(V)$ which will be denoted by $O$,
\begin{equation}
O\in \SO(V)\,.
\end{equation}
We can identify the space of bivectors with $V\oplus V$ by the map
\begin{equation}\begin{split}
\phi =(\phi_+,\phi_-) \colon {\bigwedge}^2 \R^{1,3}&\rightarrow V\oplus V\\
 W&\mapsto (\phi_+(W),\phi_-(W))=(n\llcorner W,n\llcorner \ast W)\,.
\end{split}\end{equation}
As $e^Bn=n$ and the rest of the operations is $\SO_+(1,3)$ invariant, the decomposition is equivariant to
\begin{equation}
\phi_\pm(e^BW)=O\phi_\pm(W)\,.
\end{equation}
The only vector preserved by $O$ is its axis of rotation $h=\phi_-(B)$, thus the space of preserved bivectors is given by
\begin{equation}
\phi^{-1}(-\beta\phi_-(B),\alpha\phi_-(B))=\alpha B+\beta\ast B
\end{equation}
for $\alpha$, $\beta\in \R$ arbitrary.
\end{proof}

The second result was proven in \cite{Kaminski:2019dld} (Lemma 20) but not stated in this generality:

\begin{lemma}\label{lm:Biv-2}
Suppose that $v_1,v_2,v_3$ and $e$ are linearly independent vectors in $\R^{1,3}$ and $e$ is spacelike. Then
\begin{equation}
v_i\wedge e,\quad \ast (v_i\wedge e),\ i=1,\ldots, 3
\end{equation}
are linearly independent bivectors.
\end{lemma}

\begin{proof}
Let us notice an identity for any vector $v$
\begin{equation}
e\llcorner \ast (v\wedge e)=0\,.
\end{equation}
Moreover, for an arbitrary vector $v$,
\begin{equation}
e\llcorner (v\wedge e)=0\Longleftrightarrow v=\gamma e\,.
\end{equation}
Consider a bivector
\begin{equation}
B=\sum_i \alpha^iv_i\wedge e+\ast\sum_i \beta^iv_i\wedge e\,.
\end{equation}
Suppose that $B=0$. Contracting it with $e$, we obtain
\begin{equation}
\sum_i \alpha^iv_i=\gamma e\,.
\end{equation}
Due to the independence of $v_1,v_2,v_3$ and $e$, it means that $\alpha^i=0$ for $i=1,2,3$.

Contracting $\ast B$ with $e$, on the other hand, we obtain (due to $\ast^2=(-1)$ in the Lorentzian signature)
\begin{equation}
\sum_i \beta^iv_i=\gamma' e\,.
\end{equation}
Therefore, $\beta^i=0$ for $i=1,2,3$. The linear independence is hence proven.
\end{proof}

\begin{proof}[Proof of Proposition \ref{prop:U-U} ]
Choose $a,b$ and let $c,d,e$ be the remaining vertices. We have the following identities.
\begin{align}
&\Ad_{G_{cba}}U_{da}-U_{da}=\Ad_{G_{cba}}U_{ea}-U_{ea}\,,\\
&\Ad_{G_{dba}}U_{ea}-U_{ea}=\Ad_{G_{dba}}U_{ca}-U_{ca}\,,\\
&\Ad_{G_{eba}}U_{ca}-U_{ca}=\Ad_{G_{eba}}U_{da}-U_{da}\,.
\end{align}
We introduce bivectors $V_i$ for $i\in\{c,d,e\}$:
\begin{equation}
V_c=U_{da}-U_{ea},\quad V_d=U_{ea}-U_{ca},\quad V_e=U_{ca}-U_{da}\,.
\end{equation}
They satisfy $V_c+V_d+V_e=0$ and 
\begin{equation}
\Ad_{G_{cba}}V_{c}=V_c,\quad \Ad_{G_{dba}}V_{d}=V_d,\quad \Ad_{G_{eba}}V_{e}=V_e\,.
\end{equation}
This means there exist constants $\alpha^i,\beta^i$, $i\in\{c,d,e\}$ such that (due to Lemma \ref{lm:Biv-1})
\begin{equation}
V_i=\alpha^iB_i+\beta^i\ast B_i,\quad i\in\{c,d,e\}\,,
\end{equation}
where $B_i=\gamma_{ia}\wedge\gamma_{ba}$ for $i\in\{c,d,e\}$. Therefore,
\begin{equation}
\sum_{i\in\{c,d,e\}}\alpha^iB_i+\beta^i\ast B_i=0\,.
\end{equation}
Furthermore, due to non-degeneracy of the $4$-simplex, $\gamma_{ca},\gamma_{da},\gamma_{ea}$ and $\gamma_{ba}$ are linearly independent. 
Lemma \ref{lm:Biv-2} now shows that $\alpha^i=\beta^i=0$ for $i\in\{c,d,e\}$, leading to
\begin{equation}
0=V_c=U_{da}-U_{ea},\Longrightarrow U_{da}=U_{ea}\,.
\end{equation}
As the choice of $a,b$ and $c$ vertices was arbitrary,
\begin{equation}
U_{ba}=U_{ca}\text{ for every } a,b,c \text{ distinct}.
\end{equation}
This shows that there exists $U_a$, $a=1,\ldots, 5$ such that
\begin{equation}
U_{ba}=U_a
\end{equation}
for every $b\not=a$.
\end{proof}

Let $(g_{ab})\in \cM_{\Flat}(M_3,\SLCC)$ belong to ${\mathcal L}_{\rm coh}$. Conditions for vanishing of the tangent vector (to the space of flat connections on $M_3$) $t$ projected by $\pi_{\cS_d}\iota$ for $d=1,\ldots, 5$ is given by the existence of $u_{ab}$ such that
\begin{equation}
\delta_t g_{ab}=u_{da}g_{ab}-g_{ab}u_{db}
\end{equation}
for every $(a,b)\in L_2 ^d$. By the chain rule, for $g_{cba}=g_{ac}g_{cb}g_{ba}$, we have
\begin{equation}
\delta_t g_{cba}=u_{da}g_{cba}-g_{cba}u_{da}\,.
\end{equation}
This means that
\begin{equation}
\Ad_{g_{cba}}u_{ea}-u_{ea}=\Ad_{g_{cba}}u_{da}-u_{da}
\end{equation}
for distinct $a,b,c,d,e$. We can identify $\slcc$ with the space of bivectors ${\bigwedge}^2 \R^{1,3}$ (Lie algebras of $\SLCC$ and $\SO_+(1,3)$ are identical), and the adjoint action of $\SLCC$ factorizes through $\SO_+(1,3)$. Thus, 
\begin{equation}
\Ad_{G_{cba}}U_{da}-U_{da}=\Ad_{G_{bca}}U_{d'a}-U_{d'a}
\end{equation}
By Proposition \ref{prop:U-U}, $U_{ba}=U_a$. So, using the identification of bivectors and $\slcc$, there exist $u_f$, $f\in L_1$ such that
\begin{equation}
\delta_t g_{ab}=u_{a}g_{ab}-g_{ab}u_{b}
\end{equation}
for $a$, $b$ distinct. {\Blue Theorem \ref{thm:geometric} (Theorem \ref{thm:geometric-formal}) now follows from Lemma \ref{lm:non-deg-1}.}

\section{Semiclassical analysis of the $\Lambda$-SF model}
\label{sec:critical_spinfoam}

\blue{
Having established the geometric criterion for non-degeneracy in Section \ref{sec:classical}, we now apply this framework to the concrete definition of the $\Lambda$-SF vertex amplitude $\cA_v$. This amplitude is defined as a scalar product between a boundary coherent state $\Psi_{k,\rm coh}$ and a bulk state $U_k \cZ_{k,M_3}$ within the Hilbert space of Chern-Simons theory, denoted as $\langle \Psi_{k,\rm coh}, U_k \cZ_{k,M_3}\rangle_{\CS}$. Following the formulation in \cite{Han:2021tzw, Han:2025mkc}, the amplitude can be reduced in a semiclassical approximation to an oscillatory integral suitable for stationary phase analysis.
We will consider in this paper the question of non-degeneracy of the Hessian, leaving the overall problem of applicability of the stationary phase approximation to future research.
}

\blue{
The central task of this section is to bridge the gap between the specific coordinates used to define the quantum model and the geometric moduli space analyzed in the previous section. The $\Lambda$-SF model is constructed using Fock-Goncharov and Fenchel-Nielsen (FG-FN) coordinates \cite{Fock:2003alg,Dimofte:2013lba}. To invoke the non-degeneracy result of Theorem \ref{thm:geometric}, we must demonstrate that the real Lagrangian parts defined by the action of the $\Lambda$-SF model (see Section \ref{sec:stationary_phase}) correspond, via a local diffeomorphism, to the geometric submanifolds $\cL_{\coh}$ and $\cL_{M_3}$ in the space of flat connections, which then allows us to use Proposition \ref{prop:clean-o} to prove the non-degeneracy of the Hessian.
}

\blue{
We proceed by first reviewing the Chern-Simons setup in Section \ref{sec:CS_theory}. Section \ref{sec:FN-FG-local} introduces an important lemma related to the FG-FN coordinates, used in a later context. Then, using the real Lagrangian framework from Section \ref{sec:stationary_phase} (specifically Proposition \ref{prop:clean-o}), we analyze the real Lagrangian parts of the action of the model in Section \ref{sec:Hessian-analysis}. Finally, we show in Section \ref{subsec:proof_Hessian} that their image in the symplectic space of flat connections aligns with the geometric structures of Section \ref{sec:classical}, thereby proving that the Hessian is non-degenerate for the $\Lambda$-SF model.
}

\subsection{Vertex amplitude of the $\Lambda$-SF model}
\label{sec:CS_theory}

\blue{In this section, we review some facts about the vertex amplitude defined in \cite{Han:2025mkc}, which are important for the semiclassical analysis. In particular, the action from \cite{Han:2021tzw, Han:2025mkc} will be introduced.}

The vertex amplitude is defined as a constrained partition function of complex Chern-Simons theory on  $M_3$ with gauge group $\SL(2,\bC)$.
The Chern-Simons action is
\be
S_{\CS}[A,\Ab]=\f{t}{8\pi}\int_{M_3} \tr\left[A\w\rd A+ \f23 A\w A\w A\right]+ \f{\tb}{8\pi}\int_{M_3} \tr\left[\Ab\w\rd \Ab+ \f23 \Ab\w \Ab\w \Ab\right]\,,
\label{eq:CS_action}
\ee
where \blue{$A$ is a connection and $\Ab$ is its complex conjugate,}  $t=ks$ is a complex Chern-Simons coupling constant with $k=\f{12\pi}{\ell_\p^2\gamma|\Lambda|}\in\Z_+$ being the integer Chern-Simons level and 
\begin{equation}
 s=1+i\gamma,\quad \gamma\in \R.   
\end{equation}
Here, $\gamma\in \R$ is the Barbero-Immirzi parameter. We denote $\tb$ the complex conjugate of $t$.
 At the semiclassical limit, the Planck length $\ell_{\rm p}\rightarrow 0$, hence $k\rightarrow \infty$. 
This action comes from a formal path integral of the Holst-BF action for 4D gravity with a cosmological constant $\Lambda\neq 0$ after integrating the $B$-field \cite{Haggard:2015sl}. 

This informal definition was made concrete in \cite{Han:2021tzw} then improved in \cite{Han:2025mkc} using a definition of $\SLCC$ Chern-Simons theory developed in a series of works \cite{Gaiotto:2009hg,Dimofte:2011gm,Dimofte:2011ju,Dimofte:2013lba,Dimofte:2014zga,andersen2014complex}, {\Blue  where the Hilbert space associated to the boundary $\Sigma$}, as well as the generalized state $\cZ_{k,M_3}$ corresponding to the Lagrangian submanifold defined by flat connections on $M_3$  were introduced.
The vertex amplitude can be written as
\be
\cA_v=\langle \Psi_{k,\rm coh}, U_k \cZ_{k,M_3}\rangle_{ \CS}\,.
\label{eq:def-Av}
\ee
 where $\langle\cdot,\cdot\rangle_{\CS}$ is a scalar product in Chern-Simons Hilbert space and $\Psi_{k,\rm coh}$ is a family of generalized states (labeled by $k$) introduced in \cite{Han:2021tzw, Han:2025mkc} encoding the geometry of a chosen $4$-simplex. The operators $U_k$ and $\cZ_{k,M_3}$ are parts of the Chern-Simons theory $M_3$. 
 {\Blue We will now state properties of these objects needed in our paper. The details of their definition can be found in \cite{Han:2021tzw, Han:2025mkc}.}
 
\subsubsection{Fock-Goncharov-Felsen-Nielsen coordinates}
\label{subsubsec:FGFN}

We will now describe how the Darboux coordinates can be introduced to describe flat connections on the $\Gamma_5$ graph \cite{Han:2025mkc}. \blue{They are complex coordinates as introduced in Section \ref{sec:complex-phase-space} with $s=1+i\gamma$, and here complex dimension $N=15$.} Using the terminology of \cite{Dimofte:2011ju}, we divide the surface into the so-called cusp boundary component, which consists of a disjoint sum of annuli over every edge and the so-called geodesic boundary component, which consists of a disjoint sum of 4-holed spheres. There are 10 annuli and 5 spheres. We first introduce a framing along every part of the cusp boundary component.  It is a choice of $1$-dimensional subspace of spinors for every annulus, which is preserved by parallel transport on that annulus. A framing is a choice of a 1-dimensional spinor subspace (a line) for every annulus, which is invariant under parallel transport along that annulus. A flat connection, paired with such a framing, is called a framed flat connection. The flat connection, together with the choice of a framing, defines a framed connection. There is a natural map from the space of framed connections to the space of flat $\SLCC$ connections. 

One can check that this map is a $2^{10}$ covering map on a large subset. Indeed, if the holonomy around an annulus has trace non-equal $\pm 2$, then there exist exactly two spinor eigen-subspaces preserved by this holonomy. If this is the situation for every annulus, there are exactly two choices per annulus of the framing and the map is locally a covering. 

The Fock-Goncharov construction, augmented by a choice of Fenchel-Nielsen coordinates for annuli, provides $\C^\ast$ coordinates on a dense open subset of framed connections. Taking the logarithm of these variables, we arrive (under a suitable choice provided by \cite{Han:2025mkc}) at the Darboux coordinates. We emphasize that these Darboux coordinates do not describe framed connections, but there is an infinite covering map given by exponentiation to the open dense subset of framed connections. Taking into account the further map into $\cP_\Sigma$, we arrive at the description of the phase space.

We now shortly describe the  (complex) Darboux coordinates introduced in \cite{Han:2025mkc}.
Using notations of \cite{Han:2025mkc}, the first ten elements of them are called the {\it Fenchel-Nielson coordinates} associated to the annuli of $M_3$, denoted as 
\be\label{eq:P-Q}
{\Blue \lb P_I,Q_I\rb_{I=1,\cdots,10}=\{T_{ab},L_{ab}\}_{a<b}\,.}
\ee 
where $L_{ab}$, called the (Fenchel-Nielsen) length, is the logarithm of the eigenvalue of holonomy around annulus for the chosen spinor framing at the annulus $(ab)$ and $T_{ab}$ is a conjugate twist coordinate \cite{Dimofte:2013lba}. 

The last five elements are called the {\it Fock-Goncharov coordinates} associated to the 4-holed spheres $\cS_a$, denoted as
\be\label{eq:S_a-variables}
\blue{ \lb P_{a+10}, Q_{a+10}\rb_a = \lb Y_a,X_a\rb_a\,,\quad a=1,\cdots,5\,.}
\ee
Let $\FF=\C^{30}$ be the space of the logarithmic FG-FN coordinates.
As described above, there is a map preserving the symplectic form
\begin{equation}
\pi_{FG}\colon \FF\rightarrow \cM_{\Flat}(\Sigma,\SLCC)\,,
\end{equation}
{\Blue The map is neither injective nor surjective.}

We can also describe the space $\cM^0(\cS_a,\SLCC)$ with the above coordinates. \blue{Namely, a local coordinate system is provided by}
\begin{equation}
L_{ab},\ b\not=a,\quad X_a,Y_a\,.
\end{equation}
We use here the convention that, for $a>b$, we define the length variables $L_{ab}=-L_{ba}$. 
The projection $\pi_{\cS_a}$ is given by projecting on these variables along the remaining variables in the local Darboux coordinate system on $\cP_{\Sigma}$. {\Blue It can be stated as follows:

\begin{lemma}\label{lm:FG-s-a}
Let $\Pi_a\colon \C^{30}\rightarrow \C^6$ be a projection
\begin{equation}
\Pi_a(\form{P},\vec{Q})=(\{L_{bc}\}_{ (b,c)\in L^a_2}, Y_a,X_a)\,.
\end{equation}
Then there exists a unique map $\pi^a_{FG}\colon \C^6\rightarrow \cM^0_{\Flat}(\cS_a,\SLCC)$ such that $\pi^a_{FG}\Pi_a=\pi_{\cS_a}\pi_{FG}$.
\end{lemma}

}

In the construction of the partition function, another set of Darboux coordinates, denoted as $\lb \form{\Pi}, \vec{\Phi}\rb$ is relevant.  The coordinates $\lb \form{\Pi}, \vec{\Phi}\rb$ are obtained through an ideal triangulation of $M_3$. In detail, $M_3$ can be decomposed into five ideal octahedra denoted as $\Oct(a)$, $a =1,\ldots, 5$.
Each ideal octahedron is obtained by the intersection of $M_3$ with $T_a$,  see fig.\ref{fig:oct_Ta}.  For every ideal octahedron $\Oct(a)$, we have a set of 6 Fock-Goncharov variables,
\begin{equation}\label{eq:variables-Oct}
\lb \Pi_{3a-i},\Phi^{3a-i}\rb_{i=0,1,2}\,.
\end{equation}
This division allows us to write the phase space $\FF$ as another Cartesian product:
\begin{equation}
\FF=\blue{\bigtimes_{a=1}^5}\FF_{\Oct(a)}^{\times 5}\,,
\end{equation}
where $\FF_{\Oct(a)}=\C^6$ is the phase space corresponding to $\Oct(a)$ with symplectic coordinates \eqref{eq:variables-Oct}. 

{\Blue  For our analysis, a third set of Darboux coordinates will be useful (see Section \ref{sec:Hessian-analysis}). It differs from $(\form{P},\vec{Q})$ in their first 10 elements:
\begin{equation}\label{eq:variable'}
\lb P_I',Q_I'\rb_{I=1,\cdots,10}=\{L_{ab},-T_{ab}\}_{a<b},\ \lb P_{a+10}', Q_{a+10}'\rb_a = \lb Y_a,X_a\rb_a\,,\quad a=1,\cdots,5.
\end{equation}
Any two coordinate systems $(\form{P},\vec{Q})$, $(\form{P}',\vec{Q}')$ and $(\form{\Pi},\vec{\Phi})$ are related by an affine transformation preserving the symplectic form \eqref{eq:Omega_CS}.
\begin{equation}
{\mathcal H}'(\form{P},\vec{Q})=(\form{P}',\vec{Q}'),\quad {\mathcal H}_{M_3}(\form{\Pi},\vec{\Phi})=(\form{P}',\vec{Q}')\,.
\label{eq:H'HM3}
\end{equation}
While the map ${\mathcal H}_{M_3}$ is described in \cite{Han:2025mkc}, the map ${\mathcal H}'$ can be read off by comparing \eqref{eq:variable'} with \eqref{eq:S_a-variables} and \eqref{eq:P-Q}.
}

\subsubsection{\blue{The Hilbert space, states and $U_k$ operator}}

The Hilbert space associated with $\cP_\Sigma$ is constructed by quantizing the logarithmic FG-FN Darboux coordinates. For a system with $N$ degrees of freedom where coordinates are split into positions $q_\C^i$ and momenta $p^\C_i$ for $i=1,\cdots,N$ (in the $\Lambda$-SF model, $N=15$), the corresponding Hilbert space is defined as:
\begin{equation}
\hilb_{N,k}=\otimes_{i=1}^N\hilb_k,\quad \hilb_k=L^2(\R)\otimes \C^k.
\end{equation}
The construction of the vertex amplitude involves two distinct sets of Darboux coordinates: one set $\lb \form{\Pi}, \vec{\Phi}\rb$ adapted to the ideal triangulation of the bulk $M_3$, and another set $\lb \form{P}, \vec{Q}\rb$ adapted to the boundary structure. 

More precisely, in the description using boundary coordinates $(\form{P},\vec{Q})$, the total Hilbert space factors into a product of 5 vertex Hilbert spaces and 10 edge Hilbert spaces. The boundary condition is encoded by a generalized coherent state $\Psi_{k,\coh}$ (in the sense of \cite{Han:2021tzw}). This state is not an element of the Hilbert space but a functional on a dense domain. It is constructed by imposing simplicity constraints and fixing the triangle areas in the edge Hilbert spaces to specific values $j_I$	\cite{Han:2025mkc,Pan:2025sut}:
\be
a_I=\f{3}{|\Lambda|}{\rm min}\left(\f{4\pi}{k}j_I,2\pi-\f{4\pi}{k}j_I\right)\,,\quad j_I=0,\f12,\cdots,\f{k-1}{2}\,.
\ee
These states are eigenstates of the holonomy operators corresponding to the length variables $L_{ab}$.

The Chern-Simons state (describing bulk dynamics) is defined using  $\lb \form{\Pi}, \vec{\Phi}\rb$ coordinates, which are associated with an ideal triangulation of $M_3$ into five ideal octahedra. The Hilbert space in this description is a tensor product of Hilbert spaces associated with every ideal octahedron. The state $\cZ_{k,M_3}$ is the generalized state implementing condition of bulk flatness. Explicitly, it is constructed using the quantum dilogarithm function. For details, we direct readers to \cite{Han:2021tzw, Han:2025mkc} and for general theory to \cite{Dimofte:2011gm}.

The two coordinate systems $(\form{P},\vec{Q})$ and $\lb \form{\Pi}, \vec{\Phi}\rb$ are related classically by an affine symplectic transformation ${\mathcal H}'\circ{\mathcal H}_{M_3}$. In the quantum theory, this transformation is implemented by a unitary operator $U_k$ (family of operators labelled by Chern-Simons level $k$), which maps the bulk Hilbert space to the boundary Hilbert space. This operator is a metaplectic transformation (uniquely defined up to a phase \cite{Dimofte:2011gm}) that identifies the two Hilbert spaces, thereby restoring the covariance of the quantum theory under the change of coordinates. The vertex amplitude is thus formed by mapping the bulk state $\cZ_{k,M_3}$ via  $U_k$ and contracting it with the boundary coherent state $\Psi_{k,\coh}$, giving \eqref{eq:def-Av}.

\medskip
\subsubsection{\blue{Poisson summation and asymptotic analysis}}
\label{sec:thm-1}

The vertex amplitude is not primarily given by the integral to which the method of stationary phase can be applied, due to additional discrete summations in the $\C^k$ part of the Hilbert space. To apply the stationary phase method, this sum must be converted into a continuous oscillatory integral. This is achieved using the Poisson summation formula \cite{Han:2021tzw,Han:2025mkc}. After discarding terms which decay faster than any power of $k$, the asymptotic analysis was reduced to the stationary phase method (after appropriate rescaling of variables) \cite{Han:2021tzw, Han:2025mkc}. The amplitude takes the form of an oscillatory integral governed by a total action $S_{tot}$, which consists of the following terms:
\begin{enumerate}
\item The coherent state action\footnote{See definition in Example \ref{ex:coherent}.} $S_{\coh}=\sum_{a=1}^5S_{a}(X_a)$: These terms enforce the boundary conditions. They satisfy the property $\Im S_{a}\geq 0$, with equality holding only at the specific boundary values $X_a^0$ corresponding to the geometry of the chosen 4-simplex. We denote the holomorphic derivative at these points by $Y_a^0:=\frac{s}{4\pi}\frac{\partial \Re S_a}{\partial X_a}$. 
\item The bulk action $S_{M_3}$: This is the sum of real actions associated with every ideal octahedron in the decomposition of $M_3$. This action is singular on a discrete set of points. However, the stationary points of the total action (corresponding to non-degenerate geometries) are disjoint from these singularities. \blue{We exclude these singular points from the domain of the action.}
\item  The action $S_U$: This is the generating function for the affine symplectic transformation related to the operators $U_k$ discussed above.
\end{enumerate}
The final action is given as a sum
\begin{equation}\label{eq:action-Han-Pan}
S_{tot}=-\sum_{a=1}^5\overline{S_{a}(Q_{10+a})}+S_U(\vec{Q}, \vec{\Phi})+S_{M_3}(\vec{\Phi})\,.
\end{equation}
The constraints on the edge Hilbert spaces are reflected in asymptotic analysis by fixing the length variables $L_{ab}$ to constant values.
The stationary points for the geometric configurations were analyzed in \cite{Han:2021tzw, Han:2025mkc}.

We can now state the main technical result of this paper, which serves as the formal version of Theorem \ref{thm-1}: 

\begin{theorem}\label{thm-1-formal}
The partial Hessian of action \eqref{eq:action-Han-Pan} with fixed $L_{ab}=L_{ab}^0$ variables is non-degenerate for a stationary point of the action if
\begin{equation}
m=(m_1,\ldots m_5),\quad m_a=\pi^a_{FG}\left(\{L_{bc}^0\}_{L_2^a}, Y_a^0,X_a^0\right)\, ,
\end{equation}
is a geometric boundary data.
\end{theorem}
In the rest of this section, we provide the proof of this theorem by analyzing the symplectic geometry of the FG-FN coordinates.

\subsection{\blue{FG-FN coordinates and the space of flat connections}}
\label{sec:FN-FG-local}

Recall from Section \ref{subsubsec:FGFN} that the FG-FN coordinates parametrize the space of framed flat connections, which serves as a covering space for the standard moduli space $\cM_{\Flat}(\Sigma,\SLCC)$. The relationship between the phase space $\FF$ (coordinatized by FG-FN variables) and the physical moduli space is given by the projection map $\pi_{FG}$. The following lemma establishes that this map behaves nicely in the neighborhood of the geometric configurations relevant to our main theorem. 

\begin{lemma}\label{lm:isom-FG}
Every geometric flat connection belongs to the smooth locus of $\cM_{\Flat}(\Sigma,\SLCC)$. Further, the map $\pi_{FG}\colon \FF\rightarrow \cM_{\Flat}(\Sigma,\SLCC)$ is a local diffeomorphism at every pre-image of a geometric flat connection. 
\end{lemma}

\begin{proof}
A point in  $\cM_{\Flat}(\Sigma,\SLCC)$ belongs to a smooth locus if and only if there exists no $1$-dimensional subspace preserved by parallel transport (see for example \cite{goldman1984symplectic}). We argue by contradiction. Suppose that such a subspace exists for a geometric flat connection. Let us restrict flat connection to $\cM^0_{\Flat}(\cS_a,\SLCC)$. As this connection can be put in a $\SU(2)$ gauge, the orthogonal complement of the subspace is also preserved by parallel transport. However, this means that there exists a gauge in which all group elements are in ${\rm U}(1)$, which is a contradiction with assumptions about geometric boundary data. Thus, the connection must lie in the smooth locus.

The second point now follows from the fact that the symplectic form in $\FF$
is a pull-back of the Atiyah-Bott symplectic form. Both forms are non-degenerate (in particular, the dimensions of the tangent spaces agree), thus no vector can be annihilated by $D\pi_{FG}$. Consequently, $\pi_{FG}$ is a local diffeomorphism.
\end{proof}

\subsection{\blue{Hessian analysis}}\label{sec:Hessian-analysis}

We will first simplify the analysis of the Hessian using Lemma \ref{lm:partial-Hessian-complex}. We introduce a new action $S_{tot}^{new}(\vec{Q}',\vec{Q},\vec{\Phi})$
\begin{align}\label{eq:S_tot-new-action}
S_{tot}^{new}&=-\sum_{a=1}^5\overline{S_{a}(Q_{a+10}')}+S_U(\vec{Q}, \vec{\Phi})+S_{M_3}(\vec{\Phi})+\nonumber\\
&-\sum_{I=1}^{10} \Re \left(\frac{s}{2\pi}{Q'}^I(Q^I-Q_0^I)\right)+\sum_{a=1}^{5} \Re \left(\frac{s}{2\pi}\Lambda_a(Q^{10+a}-{Q'}^{10+a})\right)
\end{align}
where $\vec{Q}_0$ are values of fixed $L_{ab}$ ($Q_0^I=-\f{4\pi i}{k}j_I$).
Using $-{Q'}^I$, $I=1,\ldots, 10$ and $\Lambda_a$, $a=1,\ldots, 5$ as Lagrange multipliers, we obtain by Lemma \ref{lm:partial-Hessian-complex} that the Hessian ${\bf H}(S_{tot}^{new})$ is non-degenerate if and only if the Hessian \eqref{eq:action-Han-Pan} is non-degenerate.

The action $S_{tot}^{new}$ can be written as
\begin{equation}
S_{tot}^{new}=-\overline{S_+}+S_{H}+S_-
\end{equation}
where we introduced
\begin{enumerate}
\item $S_+(\vec{Q}')=\sum_{a=1}^5S_a({Q'}^{10+a})-\sum_{I=1}^{10}\Re\left(\frac{s}{2\pi}{Q'}^IQ_{0}^I\right)$ with $Q_0^I=-\f{4\pi i}{k}j_I$;
\item $S_H(\vec{Q}',\vec{Q},\vec{\Phi},\form{\Lambda})=S_U(\vec{Q},\vec{\Phi})+S_{{\mathcal H}'}(\vec{Q}',\vec{Q},\form{\Lambda})$ where
\begin{equation}
S_{{\mathcal H}'}(\vec{Q}',\vec{Q},\form{\Lambda})=-\sum_{I=1}^{10}\Re\left(\frac{s}{2\pi}{Q'}^IQ^I\right)+\sum_{a=1}^{5} \Re \left(\frac{s}{2\pi}\Lambda_a(Q^{10+a}-{Q'}^{10+a})\right)\, ;
\end{equation}
\item $S_-(\vec{\Phi})=S_{M_3}(\vec{\Phi})$.
\end{enumerate}
Our first goal is to reduce the analysis of the Hessian using Proposition \ref{lm:U-asym} to a geometric question about symplectic transformations and real Lagrangian parts.

\subsubsection{\blue{Real Lagrangian part for the Chern-Simons action}}

We begin by analyzing the bulk action $S_-(\vec{\Phi})=S_{M_3}(\vec{\Phi})$ (describing Chern-Simons dynamics) and the symplectic transformation connected to the action $S_U(\vec{Q},\vec{\Phi})$. 

Recall that the action $S_U$ is obtained in the semiclassical limit from the operator $U_k$, which intertwines the bulk and boundary Hilbert spaces. Geometrically, this means $S_U$ should be the generating function for the symplectic transformation between the bulk coordinates $\lb \form{\Pi}, \vec{\Phi}\rb$ and the boundary coordinates $(\form{P},\vec{Q})$. The following lemma, proven (but not stated in this way) in \cite{Han:2021tzw}, makes this identification precise and relates the bulk real Lagrangian part to the geometric manifold $\cL_{M_3}$.

\begin{lemma}\label{lm:Han-Pan-Dimofte}
The following is true:
\begin{enumerate}
\item The action $S_U$ is quadratic, and it is the generating action for an affine symplectic transformation $\chi_{S_U}$. Under the complex polarization identification $\Iota$ (defined in \eqref{eq:Iota}), this transformation implements the change of Darboux coordinates ${\mathcal H}_{M_3}$ defined in \eqref{eq:H'HM3}:
\be
\Iota^{-1}\chi_{S_U}\Iota={\mathcal H}_{M_3}\,. 
\ee
\item Consider the real Lagrangian part ${\mathcal L}_{S_{M_3}}$ for the action $S_{M_3}$. Define its image under the coordinate transformation 
\be
\tilde{\mathcal L}_{M_3}:={\mathcal H}_{M_3}\Iota^{-1}({\mathcal L}_{S_{M_3}})\,.
\label{eq:L_M3}
\ee 
Then
\begin{equation}
\pi_{FG}(\tilde{\mathcal L}_{M_3})\subset {\mathcal L}_{M_3}\,.
\end{equation}
\end{enumerate}
\end{lemma}

\begin{proof}
The action $S_U$ is real and quadratic; the only non-trivial fact is to show that it defines an affine symplectic transformation. This is a consequence of Eq.(136)-(138) in \cite{Han:2021tzw}. 

For the second point, it is shown in \cite{Han:2021tzw} (Eq.(139)-(143) and the reasoning below) that the pre-image $\Iota^{-1}({\mathcal L}_{S_{M_3}})$ of ${\mathcal H}_{M_3}$ in the $(\form{\Pi},\vec{\Phi})$ logarithmic Fock-Goncharov coordinates corresponds exactly to the set of flat connections that extend continuously from the boundary $\Sigma$
 to the bulk $M_3$. The map $\pi_{FG}$ projects these coordinates to the moduli space $\cM_{\Flat}(\Sigma,\SL(2,\C))$, and by definition, the image of such extendable connections belongs to the Lagrangian $\cL_{M_3}$.
 \end{proof}

\begin{remark}
Both $\Iota^{-1}\left(\chi_{S_U}({\mathcal L}_{S_{M_3}})\right)$ and ${\mathcal L}_{M_3}$ are Lagrangian and $\pi_{FG}$ is a local diffeomorphism, thus in fact, the map $\pi_{FG}$ is a local diffeomorphism of this two submanifolds.
\end{remark}

\subsubsection{\blue{Real Lagrangian part for the boundary state}}

We now analyze the boundary action $S_+(\vec{Q}')$, which encodes the boundary conditions. As seen in \eqref{eq:S_tot-new-action}, this action factorizes into a sum over the 15 components of the graph (10 edges and 5 vertices):
\begin{equation}
\FF=\blue{\bigtimes_{I=1}^{15}}\FF_{I}\,,
\end{equation}
where $\FF_I=\C^2$ is the phase space of variables $(P_I',{Q'}^I)$.

For every 4-holed sphere, $I=a+10$ with $a=1,\ldots, 5$, we have a coherent state action $S_a({Q'}^{10+a})$. As established in Section \ref{sec:thm-1} (recall also Example \ref{ex:coherent}), the real Lagrangian part of a coherent state action consists of a single point fixed by the geometric data:
\begin{equation}
\Iota^{-1}({\mathcal L}^{r}_{S_{a}})=\Iota^{-1}({\mathcal L}^{\ro}_{S_{a}})=\{(Y_{a}^0,X_{a}^0)\}\,.
\end{equation}
For every annulus, $I=1,\ldots, 10$, the action is
\begin{equation}
S_I({Q'}^I)= -\Re\left(\frac{s}{2\pi}{Q'}^IQ_{0}^I\right)\,.
\end{equation}
Let us notice that the action is real and $\frac{\partial}{\partial {Q'}^I}\Re\left(\frac{s}{2\pi}{Q'}^IQ_{0}^I\right)=\frac{s}{4\pi}{Q'}^I$, thus the real Lagrangean part is equal
\begin{equation}
\Iota^{-1}\left({\mathcal L}^r_{S_I}\right)=\Iota^{-1}\left({\mathcal L}^{\ro}_{S_I}\right)=\left\{(P_I',{Q'}^I)\colon P_I=-\f{4\pi i}{k}j_I\right\}\,.
\end{equation}
This fixes momenta (to be purely imaginary and proportional to the triangle areas) while leaving the twist coordinates ($Q'_I$) unconstrained. The real Lagrangian part for the total action is a product of the corresponding real Lagrangian parts. It can be described as follows:
\begin{equation}
\Iota^{-1}\left({\mathcal L}^r_{S_{+}}\right)=\Iota^{-1}\left({\mathcal L}^{\ro}_{S_{+}}\right)=\{(\form{P}',\vec{Q}')\colon {Q'}^{a+10}=Q^{a+10}_\ast,\ a=1,\ldots, 5,\ P_I'=P_I^\ast,\ I=1,\ldots 15\}\,,
\end{equation}
where $P_I^\ast=-\f{4\pi i}{k}j_I$ for $I\leq 10$ and $Q_\ast^{a+10}=X_a^0$, $P_{a+10}^\ast=Y_a^0$, $a=1,\ldots, 5$ are fixed.

To relate it to the FG-FN coordinates $(\form{P},\vec{Q})$, we must map it through the affine transformation ${\mathcal H}'$ (defined in \eqref{eq:H'HM3}) that relates the $(\form{P}',\vec{Q}')$ system to the $(\form{P},\vec{Q})$ system. 
We introduce
\begin{equation}
\tilde{\mathcal L}_{\coh}:={{\mathcal H}'}^{-1}\left(\Iota^{-1}\left({\mathcal L}^{\ro}_{S_{+}}\right)\right)\,.
\label{eq:L_coh}
\end{equation}
The following lemma confirms that this manifold corresponds to the geometric boundary data defined in Section \ref{sec:classical}.

\begin{lemma}\label{lm:coh-FG}
The following holds
\begin{equation}
\pi_{FG}\left(\tilde{\mathcal L}_{\coh}\right)\subset {\mathcal L}_{\coh}(m)\,,
\end{equation}
where $m=(m_1,\ldots m_5)$ with $m_a=\pi^a_{FG}\left(\{L_{bc}^0\}_{L_2^a}, Y_a^0,X_a^0\right)$.
\end{lemma}

\begin{proof}
The proof is based on the fact that
\begin{equation}
\tilde{\mathcal L}_{\coh}=\{(\form{P},\vec{Q})\colon {Q}^{I}=Q^{I}_\ast,\ a=1,\ldots 15,\ P^{10+a}=Q^{10+a}_\ast,\ a=1,\ldots, 5\}\,,
\end{equation}
where $Q^I_\ast=-\f{4\pi i}{k}j_I$ for $I\leq 10$ and $Q_\ast^{a+10}=X_a^0$, $P_{a+10}^\ast=Y_a^0$, $a=1,\ldots, 5$ are fixed. This leads to conclusion that, for every point $\tilde{x}\in \tilde{\mathcal L}_{\coh}$, 
\begin{equation}
\pi_{\cS_a}(\pi_{FG}(\tilde{x}))=\pi^a_{FG}\Pi_a(\tilde{x})=\pi^a_{FG}\left(\{L_{bc}^0\}_{L_2^a}, Y_a^0,X_a^0\right)\,,
\end{equation}
where the first equality results from Lemma \ref{lm:FG-s-a}. 
From the definition \eqref{def:coh-Lamgr} of ${\mathcal L}_{\coh}(m)$, $\pi_{FG}(\tilde{x})\in {\mathcal L}_{\coh}(m)$.
\end{proof}

\begin{remark}
\label{rmk:iso}
Lemma \ref{lm:coh-FG} establishes an inclusion, but the relationship is stronger. Since $\pi_{FG}$
 is a local diffeomorphism on the ambient phase space (Lemma \ref{lm:isom-FG}), its restriction to $\tilde{\mathcal L}_{\coh}$ induces a local diffeomorphism onto $\cL_{\coh}(m)$ in the neighborhood of any geometric flat connection. Crucially, this implies that the tangent map $D\pi_{FG}$ provides an isomorphism between the tangent spaces $T_{\tilde{x}}\tilde{\mathcal L}_{\coh}$ and $T_{x}{\mathcal L}_{\coh}(m)$.
\end{remark}

\subsubsection{\blue{Action $S_{{\mathcal H}'}$ and its symplectic transformation}}

We will now show that $S_{{\mathcal H}'}$ is a generating function of an affine symplectic transformation and $\chi_{S_{{\mathcal H}'}}=\Iota{\mathcal H}'\Iota^{-1}$. Direct computation gives
\begin{align}
&\frac{\partial S_{{\mathcal H}'}}{\partial Q^I}=-\frac{s}{4\pi}{Q'}^I,\quad\frac{\partial S_{{\mathcal H}'}}{\partial {Q'}^I}=-\frac{s}{4\pi}Q^I,\quad I=1,\ldots, 10\,,\\
&\frac{\partial S_{{\mathcal H}'}}{\partial Q^{10+a}}=\frac{s}{4\pi}\Lambda_a,\quad\frac{\partial S_{{\mathcal H}'}}{\partial {Q'}^{10+a}}=-\frac{s}{4\pi}\Lambda_a,\quad a=1,\ldots, 5\,,\\
&\frac{\partial S_{{\mathcal H}'}}{\partial \Lambda_a}=\frac{s}{4\pi}(Q^{10+a}-{Q'}^{10+a})\,.
\end{align}
By Lemma \ref{lm:real-Lagr-complex} and under the condition $\frac{\partial S_{{\mathcal H}'}}{\partial \Lambda_a}=0$, $a=1,\ldots, 5$, 
\begin{equation}
P_I'=\left\{\begin{array}{ll}
-Q^I, &I=1,\ldots, 10\\ P_I &I=11,\ldots 15
\end{array}\right.,\quad
{Q'}^I=\left\{\begin{array}{ll}
P_I, &I=1,\ldots, 10\\ Q^I &I=11,\ldots 15
\end{array}\right.,\quad \Lambda_a=-P_{10+a}\,.
\end{equation}
Comparing these relations with the definition of the affine map ${\mathcal H}'$ \eqref{eq:H'HM3}, we confirm that the action generates an affine symplectic transformation $\Iota^{-1}{\mathcal H}'\Iota$.

Since $S_U$ generates $\Iota^{-1}\chi_{S_U}\Iota={\mathcal H}_{M_3}$ (Lemma \ref{lm:Han-Pan-Dimofte}) and ${\mathcal H}'$ generates $\Iota^{-1}{\mathcal H}'\Iota$, 
from properties of generating functions, $S_H=S_{{\mathcal H}'}+S_{U}$ is a generating function for affine sympletic transformation 
\begin{equation}
\chi_{S_H}=\chi_{S_{{\mathcal H}'}}\circ\chi_{S_U}=\Iota{\mathcal H}'{\mathcal H}_{M_3}\Iota^{-1}\,.
\label{eq:chi_SH}
\end{equation}
This transformation maps the bulk coordinates directly to the auxiliary boundary coordinates, allowing us to compute the intersection in the final step of the proof.

\subsection{\blue{Proof of non-degeneracy of the Hessian}}
\label{subsec:proof_Hessian}

We will now prove Theorem \ref{thm-1-formal}. According to Lemma \ref{lm:partial-Hessian-complex}, we should show that the Hessian for $S_{tot}^{new}$ \eqref{eq:S_tot-new-action} is non-degenerate. Proposition \ref{lm:U-asym} reduces this task to showing that
\begin{equation}
T_{\tilde{x}'}{\mathcal L}^{\ro}_{S_+}\cap T_{\tilde{x}'}\left(\chi_{S_H}{\mathcal L}^{\ro}_{S_-}\right)=\emptyset
\label{eq:transverse-1}
\end{equation}
for $\tilde{x}'\in {\mathcal L}_{S_+}^{\ro}\cap \left(\chi_{S_H}{\mathcal L}_{S_-}^{\ro}\right)$ corresponding to a stationary point. Let us notice that ${\mathcal L}_{S_+}^{r}={\mathcal L}_{S_+}^{\ro}$ and ${\mathcal L}_{S_-}^{r}={\mathcal L}_{S_-}^{\ro}$.

We can pull this condition back to the phase space $\FF$ (parametrized by the FG-FN coordinates) using the symplectic maps identified in the previous subsections. Recall the definitions of the submanifolds from \eqref{eq:L_coh} and \eqref{eq:L_M3}:
\begin{equation}
\tilde{\mathcal L}_{\coh}={{\mathcal H}'}^{-1}\Iota^{-1}{\mathcal L}^{\ro}_{S_+},\quad \tilde{\mathcal L}_{M_3}={{\mathcal H}'}^{-1}\Iota^{-1}\left(\chi_{S_H}{\mathcal L}^{\ro}_{S_-}\right)\,.
\end{equation}
Since $\Iota^{-1}$ and ${{\mathcal H}'}^{-1}$ are diffeomorphisms, \eqref{eq:transverse-1} is equivalent to
\begin{equation}\label{cond:intersection-up}
T_{\tilde{x}}\tilde{\mathcal L}_{\coh}\cap T_{\tilde{x}}\tilde{\mathcal L}_{M_3}=\emptyset\,,
\end{equation}
where $\tilde{x}={{\mathcal H}'}^{-1}\Iota^{-1}(\tilde{x}')\in \tilde{\mathcal L}_{\coh}\cap \tilde{\mathcal L}_{M_3}$.

Suppose that $v\in T_{\tilde{x}}\tilde{\mathcal L}_{\coh}$. Then by Lemma \ref{lm:coh-FG},
\begin{equation}
D_{\tilde{x}}\pi_{FG}(v)\in T_{x}{\mathcal L}_{\coh}(m)\,,
\end{equation}
where $x=\pi_{FG}(\tilde{x})\in {\mathcal L}_{\coh}(m)$.

Similarly, if $v\in T_{\tilde{x}}\tilde{\mathcal L}_{M_3}$, then by Lemma \ref{lm:Han-Pan-Dimofte}, 
\begin{equation}
D_{\tilde{x}}\pi_{FG}(v)\in T_{x}{\mathcal L}_{M_3}\,.
\end{equation}
Moreover, as $m$ is a geometric boundary data, $x$ is the geometric flat connection corresponding to a non-degenerate 4-simplex. Theorem \ref{thm:geometric-formal} (Theorem \ref{thm:geometric}) states that, for such geometric data, $T_{x}{\mathcal L}_{M_3}\cap T_{x}{\mathcal L}_{\coh}(m)=\emptyset$. Therefore, $D_{\tilde{x}}\pi_{FG}(v)=0$. 
Finally, Lemma \ref{lm:isom-FG} guarantees that $\pi_{FG}$ is a local diffeomorphism at $x$, meaning $D_{\tilde{x}}\pi_{FG}$ is an isomorphism (it has a trivial kernel). Consequently, $D_{\tilde{x}}\pi_{FG}(v)=0$ implies $v=0$, which shows \eqref{cond:intersection-up}. This proves that the intersection \eqref{eq:transverse-1} is trivial, and therefore the Hessian is non-degenerate.

\section{\blue{Conclusion and discussion}}

In this paper, we have shown that the Hessian obtained in the stationary phase analysis of the vertex amplitude in the $\Lambda$-SF model introduced in \cite{Han:2021tzw} and later improved in \cite{Han:2025mkc} is non-degenerate, given that the boundary condition describes the geometry of a non-degenerate 4-simplex (with spacelike tetrahedra as required in the model). 

The proof is divided into two distinct parts: Theorem \ref{thm-1} (Theorem \ref{thm-1-formal}) is devoted to the specific construction of the $\Lambda$-SF model developed in \cite{Han:2021tzw, Han:2025mkc}. The second part (Theorem \ref{thm:geometric}, stated precisely in Theorem \ref{thm:geometric-formal}) concerns properties of special submanifolds in Chern-Simons phase space. It does not involve specific details of the $\Lambda$-SF model. Typical construction of a spinfoam model is based on imposing simplicity constraints on the partition function of a TQFT. Thus, we expect that our result will find application also in other models if the TQFT involved is related to Chern-Simons theory.

The significance of this result is twofold. First, it pushes the well-definedness of the standard stationary phase approximation for the geometric sector of the model further, placing the asymptotic analysis on firmer mathematical ground. Second, it confirms that the pathological behavior observed in the Barrett-Crane model \cite{Kaminski:2013yca}, where geometric configurations possess degenerate Hessians, is absent in the $\Lambda$-SF model.

Our main theorems concern geometric boundary data and non-degenerate reconstructed $4$-simplices. 
However, the phase space of the theory allows for other possibilities, which should be investigated in order to fully understand the semiclassical limit of the $\Lambda$-SF model.
\begin{enumerate}
    \item {\bf Generalized tetrahedra}: The boundary data may correspond to ``generalized tetrahedra", defined in \cite{Haggard:2015ima}, which includes not only the tetrahedra of a geometric 4-simplex considered in this paper but two-sheeted hyperbolic tetrahedra. In this case, if the two distinct stationary points correspond to generalized non-degenerate 4-simplices, we expect our proof of transversality to extend straightforwardly to this setting, as the symplectic geometry remains largely unchanged. 
    \item {\bf Vector geometries (degenerate 4-simplices)}: A more subtle case arises when the boundary data correspond to geometric or generalized tetrahedra, but the reconstructed 4-simplex is degenerate (e.g., ``vector geometries").  The Hessian in this situation has not been analyzed in the flat EPRL nor in the $\Lambda$-SF model. In the case of the EPRL model, if there are two degenerate stationary points, one can associate with this situation a non-degenerate Euclidean $4$-simplex \cite{Barrett:2009mw}. Since the Hessian is generically non-degenerate, we conjecture that, in this case, the Hessians in both flat and curved models are also always non-degenerate. We leave it for future work.
    \item {\bf Coalescing stationary Points}: For certain non-generic boundary data, there exists only a single stationary point. This case can be obtained as a limit of the standard situation in which two distinct stationary points may approach each other and merge. At the coalescence point, the Hessian becomes degenerate, marking a breakdown of the standard stationary phase approximation. This scenario requires higher-order asymptotic techniques (\eg Airy functions) to properly capture the behavior of the amplitude.
\end{enumerate}
The question of whether degenerate configurations dominate the path integral remains a critical open problem in spinfoam gravity. If the Hessian is degenerate for a sufficiently big subset of configurations, these configurations could dominate asymptotically in the large spin limit due to their slower decay, potentially spoiling the semiclassical limit. This phenomenon is a known risk in Lorentzian models (see \eg \cite{Barrett:2002ur}). Our result confirms that the geometric sector is safe, but a full understanding of the semiclassical limit requires a comprehensive analysis of the ``degenerate geometry" sector and its contribution to the total amplitude.

Let us compare our method with that used in \cite{Kaminski:2019dld} to prove the non-degeneracy of the Hessian in the EPRL model.
 The construction therein is based directly on the notion of ``positive Lagrangian" and analysis of the Hessians. 
It would be interesting to adapt our geometric intersection method to the EPRL context. Interestingly, the geometrical interpretation of the FG-FN coordinates motivates us to view them as some generalization of twisted geometry \cite{Freidel:2010aq,Freidel:2010bw}.  
The primary challenge in this adaptation lies in the distinct phase space structures. Namely, in the $\Lambda$-SF model, we work directly on the gauge-invariant moduli space, while the EPRL model is defined on the kinematical phase space of BF theory, where gauge-invariance constraints are not yet imposed. The gauge invariance of the partition function is obtained by a further group averaging operation. Consequently, generalizing our method requires lifting the submanifold $\cL_{\rm flat}$
 from the reduced phase space to the kinematical phase space and accounting for the gauge orbits.
We leave it for future investigations. 

Let us also emphasize that, although a degenerate Hessian often signals the presence of invariant directions in the action associated with gauge symmetries (these directions are handled by gauge-fixing or integrating over the group orbit), it is not the case for the $\Lambda$-SF model. This is because our analysis is performed on the moduli space of flat connections, which are already gauge-invariant. Consequently, any degeneracy in the Hessian for the $\Lambda$-SF model would not correspond to a redundancy of variables, but rather to a degeneracy of the geometric sector.

 Finally, let us comment on the rigorous applicability of stationary phase analysis to the $\Lambda$-SF model. As shown in \cite{Han:2021tzw,Han:2025mkc}, the integral analyzed in our work is absolutely convergent. However, this is not a sufficient condition for the stationary phase method to give the right answer about the asymptotic behavior of the integral. This is because the standard theorems about stationary phase approximation assume compact integration domains, which is not the case in $\langle \Psi_{k,\rm coh}, U_k \cZ_{k,M_3}\rangle_{\CS}$. To fully validate the asymptotic expansion, one must control the contributions from the ``tails" of the integral to ensure they are sub-dominant compared to the critical point contribution. \blue{Another problematic contribution may possibly come from points where action become singular. These points should be separately considered and their contributions should be analyzed.}
 We leave the formal proof of this asymptotic validity as an open problem. 
\blue{As a first approach, numerical checks could serve as a robust test of our framework. 
An approach would be to scan a wide range of boundary data, computing the amplitude and comparing the results with contributions from stationary phase analysis using the algorithms for geometric reconstruction of critical points developed in \cite{Pan:2025sut}. Such comparison could provide numerical verification of the applicability of stationary phase analysis.}

\section*{Acknowledgements}
The authors acknowledge IQG at FAU Erlangen-N\"urnberg for the hospitality during
their visits, where work was initiated.
QP receives support from the Jumpstart Postdoctoral Program and the College of Science Research Fellowship at Florida Atlantic University, and the Shuimu Tsinghua Scholar Program of Tsinghua University. WK acknowledges financial support from a grant 2022/47/B/ST2/02735 from the Polish Science Foundation (NCN).

\noindent {\bf Data Availability Statement.} Data sharing is not applicable to this article as no new data were created or analyzed in this study.

\noindent {\bf Coflict of interest statement.} The author certiﬁes that he has no affiliations with or involvement in any organization or entity with any financial interest or non-financial interest in the subject matter discussed in this manuscript.

\noindent {\bf Ethical Statement.} This research does not involve human participants or animals, and therefore no ethical approval is required.

\noindent {\bf Informed Consent.} Not applicable. This study does not involve human participants. 

\noindent {\bf Funding.} This work was supported by Shuimu Tsinghua Scholar Program of Tsinghua University, the Jumpstart Postdoctoral Program and the College of Science Research Fellowship at Florida Atlantic University, and Grant 2022/47/B/ST2/02735 from the Polish Science Foundation (NCN). 

\noindent {\bf Open Access.} This article is distributed under the terms of the Creative Commons Attribution 4.0 International License (\href{https://creativecommons.org/licenses/by/4.0/}{https://creativecommons.org/licenses/by/4.0/}), which permits unrestricted use, distribution, and reproduction in any medium, provided you give appropriate credit to the original author(s) and the source, provide a link to the Creative Commons license, and indicate if changes were made.

\bibliographystyle{bib-style} 
\bibliography{HS.bib}

\end{document}